\documentclass[11pt,a4paper]{article}

\usepackage[english]{babel}
\usepackage[utf8]{inputenc}

\usepackage{amsmath,amsfonts,amssymb,amsthm}
\usepackage{amsthm}
\usepackage{graphicx,color}
\usepackage{boxedminipage}
\usepackage{framed}
\usepackage{thmtools}
\usepackage{thm-restate}
\usepackage{xspace}
\usepackage{cite}
\usepackage{geometry}
\usepackage{fdsymbol}

\newenvironment{subproof}[1][\proofname]{%
  \begin{proof}[#1]%
}{%
  \end{proof}%
}

 \usepackage[pdftex, plainpages = false, pdfpagelabels, 
                 bookmarks=false,
                 bookmarksopen = true,
                 bookmarksnumbered = true,
                 breaklinks = true,
                 linktocpage,
                 pagebackref,
                 colorlinks = true,  % was true
                 linkcolor = blue,
                 urlcolor  = blue,
                 citecolor = red,
                 anchorcolor = green,
                 hyperindex = true,
                 hyperfigures
                 ]{hyperref} 
 \usepackage{xifthen}
 \usepackage{tabularx}
\usepackage{tikz} 
 \usetikzlibrary{calc}

\DeclareMathOperator{\operatorClassNP}{{\sf NP}}
\newcommand{\classNP}{\ensuremath{\operatorClassNP}}
\DeclareMathOperator{\operatorClassCoNP}{{\sf coNP}}
\newcommand{\classCoNP}{\ensuremath{\operatorClassCoNP}}
\DeclareMathOperator{\operatorClassFPT}{\sf {FPT}\xspace}
\newcommand{\classFPT}{\ensuremath{\operatorClassFPT}\xspace}
\DeclareMathOperator{\operatorClassW}{{\sf W}}
\newcommand{\classW}[1]{\ensuremath{\operatorClassW[#1]}}

\DeclareMathOperator{\operatorClassXP}{{\sf XP}\xspace}
\newcommand{\classXP}{\ensuremath{\operatorClassXP}\xspace}

\newcommand{\Oh}{\mathcal{O}}

\newcommand{\bran}[1]{branchable\xspace}

\newtheorem{theorem}{Theorem}
\newtheorem{lemma}{Lemma}

\newtheorem{observation}{Observation}

\theoremstyle{definition}
\newtheorem{reduction}{Reduction Rule}[section]

\newcommand{\pname}{\textsc}
\newcommand{\ProblemFormat}[1]{\pname{#1}}
\newcommand{\ProblemIndex}[1]{\index{problem!\ProblemFormat{#1}}}
\newcommand{\ProblemName}[1]{\ProblemFormat{#1}\ProblemIndex{#1}{}\xspace}

\newcommand{\probCGClong}{\ProblemName{Clustering to Given Connectivities}}

\newcommand{\probCGC}{\ProblemName{CGC}}

\newcommand{\probWCGClong}{\ProblemName{Clustering to Given Weighted Connectivities}}
\newcommand{\probWCGC}{\ProblemName{CGWC}}

\newcommand{\probAWCGC}{\ProblemName{Annotated CGWC}}

\newcommand{\probBAWCGC}{\ProblemName{Border A-CGWC}}



\newlength{\RoundedBoxWidth}
\newsavebox{\GrayRoundedBox}
\newenvironment{GrayBox}[1]%
   {\setlength{\RoundedBoxWidth}{.93\textwidth}
    \def\boxheading{#1}
    \begin{lrbox}{\GrayRoundedBox}
       \begin{minipage}{\RoundedBoxWidth}}%
   {   \end{minipage}
    \end{lrbox}
    \begin{center}
    \begin{tikzpicture}%
       \node(Text)[draw=black!20,fill=white,rounded corners,%
             inner sep=2ex,text width=\RoundedBoxWidth]%
             {\usebox{\GrayRoundedBox}};
        \coordinate(x) at (current bounding box.north west);
        \node [draw=white,rectangle,inner sep=3pt,anchor=north west,fill=white] 
        at ($(x)+(6pt,.75em)$) {\boxheading};
    \end{tikzpicture}
    \end{center}}     

\newenvironment{defproblemx}[2][]{\noindent\ignorespaces%
                                \FrameSep=6pt%
                                \parindent=0pt%
                \vspace*{-1.5em}
                \ifthenelse{\isempty{#1}}{%
                  \begin{GrayBox}{\textsc{#2}}%                
                }{%
                  \begin{GrayBox}{\textsc{#2} parameterized by~{#1}}%  
                }
                \begin{tabular*}{\textwidth}{@{\hspace{.1em}} >{\itshape} p{1.8cm} p{0.8\textwidth} @{}}%        
            }{
                \end{tabular*}%
                \end{GrayBox}%
                \ignorespacesafterend
            }

\newcommand{\defproblema}[3]{% FJR Version
  \begin{defproblemx}{#1}
    Input:  & #2 \\
    Task: & #3
  \end{defproblemx}
}%

\begin{document}

\title{Clustering to Given Connectivities\thanks{The first author have been supported by the Research Council of Norway via the project ``CLASSIS''. The second author have been supported by the project ``ESIGMA" (ANR-17-CE40-0028). }}

\author{
Petr A. Golovach\thanks{
Department of Informatics, University of Bergen, Norway. }
\and
Dimitrios M. Thilikos\thanks{AlGCo project-team, LIRMM, Université de Montpellier, CNRS, France.}~\thanks{Department of Mathematics National and Kapodistrian University of Athens, Greece.}
}

\date{}

\maketitle

\begin{abstract}
\noindent We define a general variant of the graph clustering problem where
the criterion of density for the clusters is (high) connectivity. In \probCGClong, 
we are given an $n$-vertex graph $G$, an integer $k$, and a sequence $\Lambda=\langle \lambda_{1},\ldots,\lambda_{t}\rangle$ of positive integers and we ask whether it is possible to remove at most $k$ edges from $G$ such that the resulting connected 
components are {\sl exactly} $t$ and their corresponding edge connectivities are lower-bounded by 
the numbers in $\Lambda$.  
We prove that this problem,  parameterized by $k$, is fixed parameter tractable i.e., can be solved by an $f(k)\cdot n^{O(1)}$-step algorithm, for some function $f$
that depends only on the parameter $k$. Our algorithm
uses the recursive understanding technique that is especially adapted so to 
deal with the fact that, in out setting, we do not impose any restriction
to the connectivity demands  in $\Lambda$.  
\end{abstract}

\section{Introduction}\label{sec:intro}

Clustering,
%is closely related to unsupervised learning in pattern recognition systems [81]
%also called {\sl Data Classification}~\cite{} 
% J.M. Kleinberg, E. Tardos, Approximation algorithms for classification problems with pairwise relationships: Metric labeling and Markov random fields, Journal of the ACM 49 (5) (2002) 14–23.
 deals with  grouping the elements   
of a data set,  based on some similarity measure between them.
As a general computational procedure, clustering is fundamental in several scientific 
fields including machine learning, information retrieval, bioinformatics, data compression, and pattern recognition (see~\cite{Xu15acomp,Berkhin06asur,Wong15ashor}).
In many such  applications, data sets are organized and/or represented by graphs  that naturally express relations between entities.
A  {\em graph clustering problem} asks for a partition of 
the vertices of a graph into vertex sets, called  {\em clusters},  so that 
each cluster enjoys some desirable characteristics 
of ``density'' or ``good interconnectivity'', while having few
edges  between the clusters (see~\cite{Schaeffer07grap,Bocker13clus}
for related surveys).\medskip

\noindent{\bf Parameterizations of graph clustering problems.}
As a general problem on graphs, graph clustering has many variants. Most of them depend on the 
{\sl density} criterion that is imposed on the clusters and, in most of the cases, they are {\sf NP}-complete.
However, in many real-world instances,
one may expect that the number of edges between clusters 
is much smaller than the size of the graph.
This initiated the research for parameterized algorithms 
for graph clustering problems. Here the aim is to investigate when the problem is \classFPT (Fixed Parameter Tractable), 
when parameterized by the number $k$ of edges between clusters, i.e.  it admits a $f(k)\cdot n^{O(1)}$ step algorithm, also called $\classFPT$-algorithm (see~\cite{CyganFKLMPPS15,FlumGrohebook,Niedermeier06,DoFe99} for  textbooks on parameterized algorithms and the corresponding parameterized complexity class hierarchy). More general parameterizations may also involve  $k$ edit operations to the desired cluster property.

In the most strict sense, one may demand that all vertices
in a cluster are pair-wise connected,  i.e., they form a clique. This 
corresponds to the {\sc Cluster Deletion} problem and its more general version
{\sc Cluster Editing} -- also known as {\sc Correlation Clustering} -- where we ask for the minimum edge additions or deletions that can transform a graph to a collection of cliques. 
{\sc Cluster Editing} was introduced
by Ben-Dor, Shamir, and Yakhini in~\cite{BenDorSY99clus} 
 in the context of  computational biology and, independently, by Bansal,
Blum, and Chawla~\cite{Bansal04corr} motivated by machine learning problems related to document clustering (see also~\cite{SST04}). Algorithmic research on these problems and their variants  is extensive, see~\cite{ACN05TechReport,AroraBKSH05onno,AlonMNM06quad,Edachery99grap,GiotisG06corr,SST04}. Moreover their standard parameterizations are \classFPT and there is a long list of improvements 
on the running times of the corresponding \classFPT-algorithms ~\cite{Bocker09exac,Protti09appl,Bocker12agol,BockerBBT08afix,GrammGHN05grap,FellowsLRS07effi,Guo09amor,Chen11a2ke,BockerD11even,Cao2012clus}.

%As indicated in~\cite{} proposed that in
In most practical cases, in a good clustering, it is
not  necessary that clusters induce cliques. This gives rise to several 
difference measures of density or  connectivity.
In this direction, Heggernes et al., in~\cite{HeggernesLNPT10gene},  
introduced the {\sc $(p,q)$-Cluster Graph Recognition}
problem where clusters are cliques that 
may miss at most $p$ edges (also called $\gamma$-quasi cliques)~\cite{PattilloYB13oncl,PattilloVBB13onth}). This problem was 
 generalized in~\cite{LokshtanovM13clus},
where, given a function $\mu$ and a parameter $p$, each cluster $C$ should satisfy $\mu(C)\leq p$, and was proved to be {\sf FTP} for several instantiations of $\mu$.
% \cite{LokshtanovM13clus}. 
In~\cite{HuffnerKLN14part}, Hüffner et al. 
introduced the  {\sc Highly Connected Deletion}
problem, where each cluster $C$ should induce a highly connected graph (i.e., have edge connectivity bigger than $|C|/2$ -- see also~\cite{HartuvS00,HuffnerKS15find}) and proved that this problem is  \classFPT. Algorithmic improvements 
and variants of this problem where recently studied by Bliznets and Karpov in~\cite{BliznetsK17para}.
In~\cite{GuoKKU11edit}, Guo et al. studied the problems 
{\sc $s$-Defective Clique Editing}, {\sc Average-$s$-Plex Delection}, and {\sc $\mu$-Clique Deletion} where 
each cluster $S$ is demanded to be a clique missing $s$ edges, a graph of average degree at least $|C|-s$, or a graph with average density $s$, respectively (the two first variants are \classFPT, while this is not expected for the last one).
In~\cite{ShahinpourB13dist}  clusters are of diameter at most $s$ ($s$-clubs), in~\cite{BevernMN12appr,BalasundaramBH11cliq,GuoKNU10,MoserNS12exac} every vertex of a 
cluster should have an edge to all but at most $s-1$ other vertices of it ($s$-plexes).
In~\cite{FominKPPV14tigh}, Fomin et al. considered the case where the number of clusters to be obtained is {\sl exactly} $p$ and proved that this version is also in \classFPT.
%Other extensions or/and parametrizations  of {\sc Cluster Editing} that are \classFPT where introduced by Dörnfelder et al. in~\cite{DornfelderGKW14onth} ({\sc Consensus Clustering}) and by \cite{KomusiewiczU11} ({\sc Cluster Vertex Deletion}). 
Other Parameterizations  of {\sc  Cluster Editing}  where introduced and shown to be \classFPT   in \cite{DornfelderGKW14onth,KomusiewiczU11,BetzlerGKN11aver,WuC13bala,Bodlaender10clus}). \medskip

\noindent{\bf Our results.}
Here we adopt connectivity as a general density criterion for the clusters (following the line of~\cite{HuffnerKS15find,BliznetsK17para}). 
We study a general variant of graph clustering where
we pre-specify both the number of clusters (as done in~\cite{FominKPPV14tigh})  
but also the connectivities of the graphs induced by them.
  More specifically, we consider the following clustering problem.

\defproblema{\probCGClong}%
{An $n$-vertex graph $G$, a $t$-tuple $\Lambda=\langle \lambda_{1},\ldots,\lambda_{t}\rangle$ of positive integers, $\lambda_1\leq\ldots\leq \lambda_t$, and a nonnegative integer $k$.}%
{Decide whether there is a set $F\subseteq E(G)$ with $|F|\leq k$ such that $G-F$ has $t$ connected components $G_1,\ldots, G_t$ where each $G_i$ is edge $\lambda_i$-connected for $i\in\{1,\ldots,t\}$.}

The above problem can be seen as a generalization of the well-known \textsc{$t$-Cut} problem, asking for a partition of a graph into exactly $t$ nonempty components such that the total number of  edges between the components is at most $k$. Indeed, this problem is \probCGClong for $\lambda_1=\ldots=\lambda_t=1$.
As it was observed by Goldschmidt and Hochbaum in~\cite{GoldschmidtH94}, \textsc{$t$-Cut}  is \classNP-hard if $t$ is a part of the input. This immediately implies the \classNP-hardness of \probCGClong. Therefore, we are interested in the parameterized complexity of \probCGClong. 
%(we refer to the books~\cite{CyganFKLMPPS15,DowneyF13} for the introduction to the field). 
The main results of Goldschmidt and Hochbaum~\cite{GoldschmidtH94} is that \textsc{$t$-Cut} can be solved in time $\Oh(n^{t^2})$, that is, the problem is polynomial for any \emph{fixed} $t$. In other words, \textsc{$t$-Cut} belongs in the parameterized complexity class  \classXP when parameterized by $t$. This results was shown to be tight in the sense that we cannot expect an \classFPT algorithm for this problem unless the basic conjectures of the Parameterized Complexity theory fail by Downey et al. who proved in \cite{DowneyEFPR03} that the problem is \classW{1}-hard when parameterized by $t$. The situation changes if we parameterize the problem by $k$. By the celebrated result of Kawarabayashi and Thorup~\cite{KawarabayashiT11}, \textsc{$t$-Cut} is \classFPT when parameterized by $k$. 
%Also the variant of {\sc $t$-cut} where we demand that the number of clusters is {\sl at least} $t$, is also {\sf FPT} as proved in~\cite{DowneyEFPR03cutt}. 
%
%
%We also consider the following problem:
%
%
%
%\defproblema{\probCGSlong}%
%{An $n$-vertex  graph $G$, an integer partition $n_{1},\ldots,n_{t}$ of $n$, and a nonnegative integer $k$.}%
%{Decide whether there is a set $F\subseteq E(G)$ with $|F|\leq k$ such that $G-F$
%is the disjoint union of $t$ cliques of sizes  $n_{1},\ldots,n_{t}$.}
%Notice that \probCGSlong is a special case of \probCGClong. Indeed we set $\lambda_{i}=n_{i}-1, i\in\{1,\ldots,t\}$ and we rearrange $\lambda_{i}$'s so that 
%$\lambda_1\leq\ldots\leq \lambda_t$. 

In this paper, we prove that \probCGClong is \classFPT when parameterized by $k$. Actually we consider  the edge weighted version of the problem (see Page~\pageref{pageweight}) where 
the parameter is the sum of the weights of the edges between clusters 
and the edge connectivities are now defined in terms of edge weighted cuts.
 % by showing the following theorem.
%
%\begin{theorem}\label{thm:main}
%\probCGClong can be solved in time $2^{??}\cdot n^{\Oh(1)}$ for $n$-vertex graphs.
%\end{theorem}
%
%\noindent 
For our proofs we follow the \emph{recursive understanding} technique introduced by Chitnis et al.~\cite{ChitnisCHPP16} (see also~\cite{GroheKMW11find})
combined with the \emph{random separation} technique introduced by  Cai, Chan and Chan in~\cite{CaiCC06}.
Already in~\cite{ChitnisCHPP16}, Chitnis et al. demonstrated that this technique is a powerful tool for the design of  \classFPT-algorithms for various cut problems.
  This technique was further developed by Cygan et al. in~\cite{CyganLPPS14} for proving that the \textsc{Minimum Bisection} problem is \classFPT (see also~\cite{KimPST17para,LokshtanovRSZ18redu,FominLPSW17full,GolovachHLM17find} for other recent applications of this technique  on the design of \classFPT-algorithms). Nevertheless, we stress  that for \probCGClong the application  for the recursive understanding technique becomes quite 
  nonstandard and demands additional work due to the fact that neither $t$ nor the connectivity constraints $\lambda_1,\ldots,\lambda_t$ are  restricted by any constant or any function of the parameter $k$. 
    Towards dealing with the diverse  connectivities, we  deal with  special annotated/weighted versions 
of the problem and introduce adequate connectivity mimicking encodings 
in order to make recursive understanding possible. 

\paragraph{Paper organization.} In Section~\ref{sec:defs} we give some preliminary definitions. In Section~\ref{sec:connect} we present 
the main part of our algorithm for the special case  where the input graph is connected. For this, we introduce all concepts
and results that support the applicability of the  recursive understanding  technique. We stress that, at this point, the connectivity assumption  is important as this makes it easier to control the  diverse connectivities 
of the clusters. In Section~\ref{sec:general} we deal with the general non-connected case. The algorithm in the general  case is based on a series of observations on the way connectivities are  distributed in the connected components of $G$.
% In Section~\ref{sec:nokern}, we complement our main result by proving that \probCGClong, parameterized by $k,$ does not admit a polynomial kernel unless $\classNP\subseteq\classCoNP/{\rm poly}$.
Finally, in Section~\ref{sec:concl}, we provide some conclusions, open problems, and further directions of research.

\section{Preliminaries}\label{sec:defs}
We consider  finite undirected simple graphs. We use $n$ to denote the number of vertices and $m$ the number of edges of the considered graphs unless it creates confusion. 
 For $U\subseteq V(G)$, we write $G[U]$ to denote the subgraph of $G$ induced by $U$. 
We write $G-U$ to denote the graph $G[V(G)\setminus U]$; for a single-element $U=\{u\}$, we write $G-u$. Similarly, for a set of edges $S$, $G-S$ denotes the graph obtained from $G$ by the deletion of the edges of $S$; we write $G-e$ instead of $G-\{e\}$ for single-element sets.
A set of vertices $U\subseteq V(G)$ is \emph{connected} if $G[U]$ is a connected graph.
For a vertex $v$, we denote by $N_G(v)$ the \emph{(open) neighborhood} of $v$ in  $G$, i.e., the set of vertices that are adjacent to $v$ in $G$. For a set $U\subseteq V(G)$, $N_G(U)=(\bigcup_{v\in U}N_G(v))\setminus U$. 
We denote by $N_G[v]=N_G(v)\cup\{v\}$ the \emph{closed neighborhood} of $v$; respectively, $N_G[U]=\bigcup_{v\in U}N_G[v]$.
The \emph{degree} of a vertex $v$ is $d_G(v)=|N_G(v)|$. 
For disjoint subsets $A,B\subseteq V(G)$, $E_G(A,B)$ denotes the set of edges with one end-vertex in $A$ and the second in $B$.
We may omit the subscript if it does not create confusion, i.e., the considered graph is clear from the context.   
 A set of edges $S\subseteq E(G)$ of a connected graph $G$ is an \emph{(edge) separator} if $G-S$ is disconnected.
%For two vertices $u,v\in V(G)$, $S$ is a $(u,v)$-separator if $u$ and $v$ are in distinct components of $G-S$.
For two disjoint subsets $A,B\subseteq V(G)$, $S\subseteq E(G)$ is an $(A,B)$-separator if $G-S$ has no $(u,v)$-path with $u\in A$ and $v\in B$.
Recall (see, e.g.,~\cite{Diestel12}) that if $S$ is an inclusion minimal $(A,B)$-separator, then $S=E(A',B')$ for some partition $(A',B')$ of $V(G)$ with $A\subseteq A'$ and $B\subseteq B'$.
If $A=\{a\}$ and/or $B=\{b\}$, we say that  $S$ is an $(a,b)$-separator or $(a,B)$-separator respectively.
Let $k$ be a positive integer. A graph $G$ is \emph{(edge) $k$-connected} if for every $S\subseteq E(G)$ with $|S|\leq k-1$, $G-S$ is connected, that is, $G$ has no separator of size at most $k-1$. Since we consider only edge connectivity, whenever we say that a graph $G$ is $k$-connected, we mean that $G$ is edge $k$-connected. Similarly, whenever we mention a separator, we mean an edge separator.

\paragraph{Edge weighted graphs.}
For technical reasons, it is convenient for us to work with edge weighted graphs. Let $G$ be a graph and let $w\colon E(G)\rightarrow \mathbb{N}$ be an (edge) weight function. Whenever we say that $G$ is a weighted graph, it is assumed that an edge weight function is given and we use $w$ to denote weights throughout the paper. For a set of edges $S$, ${\sf w}(S)=\sum_{e\in S}{\sf w}(e)$. 

For disjoint subsets $A,B\subseteq V(G)$, $w_G(A,B)={\sf w}(E(A,B))$. We say that $G$ is  weight $k$-connected if for every $S\subseteq E(G)$ with ${\sf w}(S)\leq k-1$, $G-S$ is connected.
We denote by $\lambda^{\sf w}(G)$ the \emph{weighted connectivity} of $G$, that is, the maximum value of $k$ such that $G$ is weight $k$-connected; we assume that every graph is weight $0$-connected and for the single-vertex graph $G$, $\lambda^{\sf w}(G)=+\infty$.  
For disjoint subsets $A,B\subseteq V(G)$,  $\lambda_G^{\sf w}(A,B)=\min\{{\sf w}(S)\mid S\text{ is an }(A,B)\text{-separator}\}$.
%We wright $\lambda^{\sf w}(G,A,B)$ if we have to specify the considered graph. 
We say that an $(A,B)$-separator $S$ is \emph{minimum} if ${\sf w}(S)=\lambda_G^{\sf w}(A,B)$. 
For two vertices $u,v\in V(G)$, $\lambda^{\sf w}(u,v)=\lambda^{\sf w}(\{u\},\{v\})$ and we assume that $\lambda^{\sf w}(u,u)=+\infty$. 
Similarly, for a set $A$ and a vertex $v$, we write $\lambda_G^{\sf w}(A,v)$ instead of $\lambda_G^{\sf w}(A,\{v\})$.
Clearly, $\lambda^{\sf w}(G)=\min\{\lambda_G^{\sf w}(u,v)\mid u,v\in V(G)\}$. %(see, e.g.,~\cite{Diestel12}).
We can omit the subscript if it does not create confusion.

Let $U\subseteq V(G)$. We say that the weighted graph $G'$ is obtained from $G$ by the \emph{weighted contraction} of $U$ if it is constructed as follows: we delete the vertices of $U$ and replace it by a single vertex $u$ that is made adjacent to every $v\in V(G)\setminus U$ adjacent to a vertex of $U$ and the weight of $uv$ is defined as $\sum_{xv\in E(G),~x\in U}{\sf w}(xv)$.
Note that we do not require $G[U]$ be connected. For an edge $uv$, the weighted contraction of $uv$ is the weighted contraction of the set $\{u,v\}$.

Let $\alpha=\langle \alpha_1,\ldots,\alpha_t\rangle$ where $\alpha_i\in\mathbb{N}\cup\{+\infty\}$ for $i\in\{1,\ldots,t\}$ and  $\alpha_1\leq\ldots\leq \alpha_t$.
We call the \emph{variate} of $\alpha$ the set of distinct elements of $\alpha$ and denote it by $\mathbf{var}(\alpha)$. Let also  $\beta=\langle \beta_1,\ldots,\beta_s\rangle$ where $\beta_i\in\mathbb{N}\cup\{+\infty\}$ for $i\in\{1,\ldots,s\}$ and  $\beta_1\leq\ldots\leq \beta_t$. We say that $\gamma=\alpha+\beta$ is the \emph{merge} of $\alpha$ and $\beta$ if $\gamma$ is the $(t+s)$-tuple obtained by sorting the elements of $\alpha$ and $\beta$ in the increasing order. We denote by $r\alpha$ the merge of $r$ copies of $\alpha$.
If $s=t$, we write $\alpha\preceq\beta$ if $\alpha_i\leq\beta_i$ for $i\in\{1,\ldots,t\}$.

Now we state the weighted variant of \probCGClong.

\label{pageweight}
\defproblema{\probWCGClong (\probWCGC)}%
{A weighted graph $G$ with an edge weight function $w\colon E(G)\rightarrow \mathbb{N}$,
 a $t$-tuple $\Lambda=\langle \lambda_{1},\ldots,\lambda_{t}\rangle$, where $\lambda_i\in\mathbb{N}\cup\{+\infty\}$ for $i\in\{1,\ldots,t\}$ and  $\lambda_1\leq\ldots\leq \lambda_t$, and a nonnegative integer $k$.}%
{Decide whether there is a set $F\subseteq E(G)$ with ${\sf w}(F)\leq k$ such that $G-F$ has $t$ connected components $G_1,\ldots, G_t$ where each $\lambda^{\sf w}(G_i)\geq \lambda_i$ for $i\in\{1,\ldots,t\}$.}

It is convenient to  allow $\lambda_i=+\infty$, because we assume that the (weighted) connectivity of the single-vertex graph is $+\infty$. 
Clearly, \probCGClong is \probWCGC for the case ${\sf w}(e)=1$ for $e\in E(G)$. For a set $F\subseteq E(G)$ with ${\sf w}(F)\leq k$ such that $G-F$ has $t$ connected components $G_1,\ldots, G_t$ where each $\lambda^{\sf w}(G_i)\geq \lambda_i$ for $i\in\{1,\ldots,t\}$, we say that $F$ is a \emph{solution} for \probWCGC.

\section{\probWCGClong for connected input graphs}\label{sec:connect}
In this section we show that \probWCGC is \classFPT when parameterized by $k$ if the input graph is connected. We prove the following theorem that is used as a building block for the general case.

\begin{theorem}\label{thm:connect}
\probWCGC can be solved in time   $2^{2^{2^{2^{2^{2^{\Oh(k)}}}}}}\cdot n^{\Oh(1)}$
 if the input graph is connected.
\end{theorem}

The remaining part of the section contains the proof of this theorem. In Subsection~\ref{sec:aux} we give some additional definitions and state auxiliary results. In Subsection~\ref{sec:high} we consider the case when the input graph is highly connected in some sense, and in Subsection~\ref{sec:proof} we compete the proof of Theorem~\ref{thm:connect}.

\subsection{Auxiliary results}\label{sec:aux}
To solve \probWCGC for connected graphs, we use the recursive understanding technique introduced by Chitnis et al. in~\cite{ChitnisCHPP16}. Therefore, we need notions that are specific to this technique and some results established by Chitnis et al.~\cite{ChitnisCHPP16}. Note that we adapt the definitions and the statements of the results for the case of edge weighted graphs.

\paragraph{Weighted good edge separations.}
Let $G$ be a connected weighted graph with an edge weight function $w\colon E(G)\rightarrow\mathbb{N}$. Let also $p$ and $q$ be positive integers. A partition $(A,B)$ of $V(G)$ is called a \emph{$(q,p)$-good edge separation} if 
\begin{itemize}
\item $|A|>q$ and $|B|>q$,
\item ${\sf w}(A,B)\leq p$,
\item $G[A]$ and $G[B]$ are connected.
\end{itemize}
%It is said that $G$ is \emph{$(q,p)$-unbreakable} if there is no $(q,p)$-good edge separation. 

\begin{lemma}[\cite{ChitnisCHPP16}]\label{lem:good}
There exists a deterministic algorithm that, given a weighted connected graph $G$ along with positive integers $p$ and $q$, in time $2^{\Oh(\min\{p,q\}\log(p+q))}\cdot n^3\log n$ either finds a $(q,p)$-good edge separation or correctly concludes $G$ that no such separation exists.
\end{lemma}

Notice that Chitnis et al.~\cite{ChitnisCHPP16} proved Lemma~\ref{lem:good} for the unweighted case but the proof could be easily rewritten for the weighted case as the authors of~\cite{ChitnisCHPP16} point themselves. 

It is said that $G$ is  \emph{$(q,p)$-unbreakable}  if for any partition $(A,B)$ of $V(G)$ such that ${\sf w}(A,B)\leq p$, it holds that $|A|\leq q$ or $|B|\leq q$. We use the easy corollary of Lemma~\ref{lem:good}. The lemma for the unweighted case was stated and proved in~\cite{FominGLS16} but we give the proof here for completeness.

 \begin{lemma}\label{lem:unbreak}
There exists a deterministic algorithm that, given a weighted connected graph $G$ along with positive integers $p$ and $q$, in time $2^{\Oh(\min\{p,q\}\log(p+q))}\cdot n^3\log n$ either finds a $(q,p)$-good edge separation or correctly concludes that $G$ is $(pq,p)$-unbreakable.
\end{lemma}

\begin{proof}
Let $G$ be a weighted connected graphs with an edge weight function $w\colon E(G)\rightarrow\mathbb{N}$. We use Lemma~\ref{lem:good} to find a $(q,p)$-good edge separation. If the algorithm returns such a separation, we return it and stop. Assume that a $(q,p)$-good edge separation does not exist. We claim that $G$ is $(pq,p)$-unbreakable. 
To obtain a contradiction, assume that $(A,B)$ is a partition of $V(G)$ such that $|A|>pq$, $|B|>pq$ and ${\sf w}(A,B)\leq p$. Because ${\sf w}(A,B)\leq p$, we have that $|E(A,B)|\leq p$ and, therefore, each of the sets $A$ and $B$ contains at most $p$ end-vertices of $E(A,B)$. Since $G$ is a connected graph, we obtain that there is a component of $G[A]$ with at least $q$ vertices. Let $X$ be the set of vertices of this component. Similarly, there is a component of $G[B]$ whose set of vertices $Y$ contains at least $q$ vertices. It follows that there is a partition $(A',B')$ of $V(G)$ with $X\subseteq A'$ and $Y\subseteq B'$ such that $G[A']$ and $G[B']$ are connected and ${\sf w}(A',B')\leq p$, but this contradicts the assumption that there is no  $(q,p)$-good edge separation of $G$.
\end{proof}

\paragraph{Mimicking connectivities by cut reductions.}
Let $r$ be a nonnegative integer. A pair $(G,\mathbf{x})$, where $G$ is a graph and $\mathbf{x}=\langle x_1,\ldots,x_r\rangle$ is a $r$-tuple of distinct vertices of $G$.
% such that each component of $G$ contains at least one vertex of $\mathbf{x}$, 
is called a \emph{$r$-boundaried graph} 
or simply a \emph{boundaried graph}. Respectively, $\mathbf{x}=\langle x_1,\ldots,x_r\rangle$ is a \emph{boundary}. Note that a boundary is an ordered set. Hence, two $r$-boundaried graphs that differ only by the order of the vertices in theirs boundaries are distinct. Still, we can treat $\mathbf{x}$ as a set when the ordering is irrelevant. 
Observe also that a boundary could be empty.  Slightly abusing notation, we may say that $G$ is a ($r$-) boundaried graph assuming that a boundary is given. We say that $(G,\mathbf{x})$ is a \emph{properly} boundaried graph if the vertices of $\mathbf{x}$ are pairwise nonadjacent and each component of $G$ contains at least one vertex of $\mathbf{x}$.

Two $r$-boundaried  weighted graphs $(G_1,\mathbf{x}^{(1)})$ and $(G_2,\mathbf{x}^{(2)})$, where $\mathbf{x}^{(h)}=\langle x_1^{(h)},\ldots,x_r^{(h)}\rangle$ for $h=1,2$, are \emph{isomorphic} if there is an isomorphism of $G_1$ to $G_2$ that maps each $x_i^{(1)}$ to $x_i^{(2)}$ for $i\in\{1,\ldots,r\}$ and each edge is mapped to an edge of the same weight. 

Let  $(G_1,\mathbf{x}^{(1)})$ and $(G_2,\mathbf{x}^{(2)})$ be  $r$-boundaried graphs with $\mathbf{x}^{(h)}=\langle x_1^{(h)},\ldots,x_r^{(h)}\rangle$ for $h=1,2$, and assume that $(G_2,\mathbf{x}^{(2)})$ is a properly boundaried graph. 
%Assume that $\{x_1^{(2)},\ldots,x_r^{(2)}\}$ is an independent set of $G_2$.
 We define the \emph{boundary sum} $(G_1,\mathbf{x}^{(1)})\oplus_b(G_2,\mathbf{x}^{(2)})$ (or  simply $G_1\oplus_b G_2$) as the (non-boundaried) graph obtained by taking disjoint copies of $G_1$ and $G_2$ and identifying $x_i^{(1)}$ and $x_i^{(2)}$ for each $i\in\{1,\ldots,r\}$. Note that the definition is not symmetric as we require that 
$(G_2,\mathbf{x}^{(2)})$ is a properly boundaried graph  and we have no such a restriction for $x_1^{(1)},\ldots,x_r^{(1)}$.

Let $\mathbf{X}=(X_1,\ldots,X_p)$ and $\mathbf{Y}=(Y_1,\ldots,Y_q)$ be two partitions of a set $Z$. We define the \emph{product} $\mathbf{X}\times\mathbf{Y}$  of $\mathbf{X}$ and $\mathbf{Y}$ as the partition of $Z$ obtained from $\{X_i\cap Y_j\mid 1\leq i\leq p,~1\leq j\leq q\}$ by the deletion of empty sets.  For partitions $\mathbf{X}^1,\ldots,\mathbf{X}^r$ of $Z$, we denote their {\em consecutive product} as $\prod_{i=1}^r\mathbf{X}^i$. The following observation is useful for us.

\begin{observation}\label{obs:prod}
If $\mathbf{X}$ and $\mathbf{Y}$ are partitions of a set $Z$, then the partition $\mathbf{X}\times \mathbf{Y}$ contains at most $|\mathbf{X}||\mathbf{Y}|$ sets. Furthermore, if $Y=(\{y_1\},\ldots,\{y_r\},Z\setminus\{y_1,\ldots,y_r\} )$ for some $y_1,\ldots,y_r\in Z$, then $\mathbf{X}\times \mathbf{Y}$ contains at most $|\mathbf{X}|+r$ sets. 
\end{observation}

Let $(H,\mathbf{x})$ be a connected properly $r$-boundaried weighted graph. Let $p$ be a positive integer or $+\infty$.  Slightly abusing notation we consider here $\mathbf{x}$ as a set. 
We construct the partition $\mathbf{Z}$ of $V(H)$ as follows. For $X\subseteq \mathbf{x}$, denote $\overline{X}=\mathbf{x}\setminus X$.
\begin{itemize}
\item For all distinct pairs $\{X,\overline{X}\}$ for nonempty
$X\subset\mathbf{x}$, find a minimum weight $(X,\overline{X})$-separator $S_X=E(Y_X^1,Y_X^2)$ where $(Y_X^1,Y_X^2)$ is a partition of $V(H)$, $X\subseteq Y_X^1$ and $\overline{X}\subseteq Y_X^2$.
%\item For $X=\mathbf{x}$, find a minimum weight $(X,v)$-separator  $S_{X}=E(Y_{X}^1,Y_{X}^2)$ for all $v\in V(H)\setminus X$ where $(Y_{X}^1,Y_{X}^2)$ is a partition of $V(H)$, $X\subseteq Y_{X}^1$.
\item For every $v\in V(H)\setminus\mathbf{x}$, find a minimum weight $(\mathbf{x},v)$-separator $S(v)$. Find $v^*\in V(H)\setminus \mathbf{x}$ such that $w(S(v^*))=\min\{w(S(v))\mid v\in V(H)\setminus \mathbf{x}\}$ and let $S(v^*)=E(Y_{\mathbf{x}}^1,Y_{\mathbf{x}}^2)$ where $(Y_{\mathbf{x}}^1,Y_{\mathbf{x}}^2)$ is a partition of $V(H)$ and $\mathbf{x}\subseteq Y_{\mathbf{x}}^1$.
\item Construct the following partition of $V(H)$:
\begin{equation}\label{eq:prod}
\!\!\!\!\!\mathbf{Z}=\big(Z_1,\ldots,Z_{h})=\big(\prod_{\text{distinct~}\{X,\overline{X}\}\atop\emptyset\neq X\subset \mathbf{x}}(Y_X^1,Y_X^2)\Big)\times (Y_{\mathbf{x}}^1,Y_{\mathbf{x}}^2)\times(\{x_1\},\ldots,\{x_r\},V(H)\setminus\mathbf{x}).
\end{equation}
\end{itemize}
We construct $H'$ by performing the weighted contraction of the sets of $\mathbf{Z}$. Then for each edge $uv$ of $H'$ with ${\sf w}(uv)>p$, we set ${\sf w}(uv)=p$, that is, we truncate the weights by $p$.
Notice that because the partition $(\{x_1\},\ldots,\{x_{r}\},V(H)\setminus\mathbf{x})$ is participating the product in the right part of (\ref{eq:prod}), we have that  $\{x\}\in \mathbf{Z}$ for each $x\in \mathbf{x}$, that is, the elements of the boundary are not contracted, and this is the only purpose of this partition in (\ref{eq:prod}).
We say that $H'$ is obtained from $(H,\mathbf{x})$ by  the \emph{cut reduction with respect to $p$}. Note that $H'$ is not unique as the construction depends on the choice of separators.
It could be observed that we construct a \emph{mimicking network} representing cuts of $H$~\cite{HagerupKNR98,KhanR14} (see also~\cite{KhanR14onmi}).

We extend this definition for disconnected graphs. Let $(H,\mathbf{x})$ be a  properly $r$-boundaried weighted graph and let $p$ be a positive integer or $+\infty$. Denote by $H_1,\ldots,H_s$ the components of $H$ and let $\mathbf{x}^i=\mathbf{x}\cap V(H_i)$ for $i\in\{1,\ldots,s\}$. Consider the boundaried graphs $(H_i',\mathbf{x}^i)$ obtained from $(H_i,\mathbf{x})$ by  cut reduction with respect to $p$ for $i\in\{1,\ldots,s\}$. We say that $(H',\mathbf{x})$, that is obtained by taking the union of $(H_i',\mathbf{x}^i)$ for $i\in\{1,\ldots,s\}$, is obtained by  the \emph{cut reduction with respect to $p$}.

The crucial property of $H'$ is that it keeps the separators of $H$ that are essential for the connectivity.

\begin{lemma}\label{lem:cut-essential}
Let $(H,\mathbf{x})$ be a properly $r$-boundaried weighted graph, 
and let $p\in\mathbb{N}\cup\{+\infty\}$. Let also $(F,\mathbf{y})$ be an $r$-boundaried weighted graph and assume that  $G=(F,\mathbf{y})\oplus_b(H,\mathbf{x})$ is connected.
Then for an $r$-boundaried weighted graph $H'$ obtained from $H$ by the cut reduction with respect to $p$, it holds that if $\lambda^{\sf w}(G)\leq p$, then $\lambda^{\sf w}(G)=\lambda^{\sf w}(G')$ where $G'=(F,\mathbf{y})\oplus_b(H',\mathbf{x})$.
\end{lemma}

\begin{proof}
Let  $\lambda^{\sf w}(G)\leq p$. 

Consider an arbitrary connected weighted graph $G$ and $U\subseteq V(G)$. Let $G'$ be the graph obtained from $G$ by the weighted contraction of $U$. Recall that each edge $e=uv$ of $G'$ is either an edge of $G$ and has the same weight or if, say, $u$ is obtained form $U$ by the contraction, then  the weight of $e$ in $G'$ is $\sum_{xv\in E(G),~x\in U}{\sf w}(xv)$. Then for  
 any separator $S'$ of $G'$, there is the separator $S$ of $G$ such that $S'$ is obtained from  $S$ by the contraction, and $S$ and $S'$ have the same weight.
Therefore, $\lambda^{\sf w}(G')\geq \lambda^{\sf w}(G)$. For $G'=(F,\mathbf{y})\oplus_b(H',\mathbf{x})$, we have that $G'$ is obtained from the connected graph $G=(F,\mathbf{y})\oplus_b(H,\mathbf{x})$ by the weighted constructions of some sets and by truncations of the weights exceeding $p$. This implies that if $G'$ has a separator of weight $c< p$, then $G$ has a separator of weight $c$. Hence, $\lambda^{\sf w}(G)\leq\lambda^{\sf w}(G')$.  

To show the opposite inequality, denote by $H_1,\ldots,H_s$ the components of $H$, $\mathbf{x}^i=\mathbf{x}\cap V(H_i)$ for $i\in\{1,\ldots,s\}$. To simplify notation, we assume without loss of generality that $\mathbf{y}=\mathbf{x}$ using the fact that the vertices of $\mathbf{x}$ are identified with the corresponding verticres of $\mathbf{y}$ in $G$. 

Let $S$ be a separator in $G$ of minimum weight. Since $\lambda^{\sf w}(G)\leq p$, ${\sf w}(S)\leq p$. Because $S$ is minimum, $S$ is an inclusion minimum separator of $G$ and, therefore, $S=E_G(A,B)$ for   some partition $(A,B)$ of $V(G)$ where $G[A]$ and $G[B]$ are connected.  Assume without loss of generality that $A\cap V(F)\neq\emptyset$.

If $A\subseteq V(F)\setminus \mathbf{x}$, then  $S=E_F(A,V(F)\setminus A)$ and, therefore, $S$ is a separator in $G'$. This implies that $\lambda^{\sf w}(G')\leq {\sf w}(S)=\lambda^{\sf w}(G)$.
Assume that this is not the case, that is, $A\cap \mathbf{x}\neq\emptyset$. 

Suppose that there is $i\in\{1,\ldots,s\}$ such that $\mathbf{x}^i\subseteq A$ and $V(H_i)\cap B\neq\emptyset$. 
We obtain that there is $v\in V(H_i)\setminus \mathbf{x}^i$ such that $\lambda^w_{H_i}(\mathbf{x}^i,v)\leq {\sf w}(S)$. Then by the construction of $H_i'$, there is $u\in V(H_i')\setminus\mathbf{x}^i$ such that 
$\lambda^w_{H_i'}(\mathbf{x}^i,u)\leq \lambda^w_{H_i}(\mathbf{x}^i,v)\leq {\sf w}(S)$. This implies that $\lambda^{\sf w}(H')\leq {\sf w}(S)=\lambda^{\sf w}(G)$.

Assume from now that for every $i\in\{1,\ldots,s\}$, either $V(H_i)\subseteq A$ or $\mathbf{x}^i\cap A\neq\emptyset$ and  $\mathbf{x}^i\cap B\neq\emptyset$.
Let $I=\{i\mid 1\leq i\leq s, \mathbf{x}^i\cap A\neq\emptyset,\mathbf{x}^i\cap B\neq\emptyset\}$. 
Observe that
\begin{equation}\label{eq:one}
E_G(A,B)=E_F(A\cap V(F),B\cap V(F))\cup(\bigcup_{i\in I}E_{H_i}(A\cap V(H_i),B\cap V(H_i))).
\end{equation}
For $i\in I$, let $X_i=\mathbf{x}^i\cap A$ and $Y_i=\mathbf{x}^i\cap B$. Note that 
\begin{equation}\label{eq:two}
{\sf w}(A\cap V(H_i),B\cap V(H_i))\geq \lambda^w_{H_i}(X_i,Y_i)\geq  \lambda^w_{H_i'}(X_i,Y_i)
\end{equation}
 where the last inequality follows from the definition of $H_i'$. Let $(A_i,B_i)$ be the partition of $V(H_i')$ such that 
$w_{H_i'}(A_i,B_i)=\lambda^w_{H_i}(X_i,Y_i)$ for $i\in I$.  Let also $A'=(A\cap V(F))\cup(\bigcup_{i\in I})A_i$ and $B'=(B\cap V(F))\cup(\bigcup_{i\in I}B_i)$.
Combining (\ref{eq:one}) and (\ref{eq:two}), we conclude
that 
$$
\lambda^{\sf w}(G)=w_G(A,B)\geq w_F(A\cap V(F),B\cap V(F))+\sum_{i\in I}w_{H_i'}(A_i,B_i)=w_{G'}(A',B')\geq\lambda^{\sf w}(G').
$$
\end{proof}

Using Lemma~\ref{lem:cut-essential} we show the following.

\begin{lemma}\label{lem:replacement}
Let $(H,\mathbf{x})$ be a properly $r$-boundaried weighted graph, 
and let $p\in\mathbb{N}\cup\{+\infty\}$ and $t\in\mathbb{N}$. Let also $(F,\mathbf{y})$ be an $r$-boundaried weighted graph, and let $G=(F,\mathbf{y})\oplus_b(H,\mathbf{x})$.
Then for an $r$-boundaried weighted graph $H'$ obtained from $H$ by the cut reduction with respect to $p$ and $t$ positive integers $\lambda_1,\ldots,\lambda_t\leq p$
it holds that  $G$ has exactly $t$ components and they have the connectivities $\lambda_1,\ldots,\lambda_t$ respectively if and only if the same holds for 
 $G'=(F,\mathbf{y})\oplus_b(H',\mathbf{x})$, that is, $G'$ has $t$ components and they have the connectivities $\lambda_1,\ldots,\lambda_t$ respectively.
\end{lemma}

\begin{proof}
%Let $H_1,\ldots,H_s$ be the components of $H$. 
To simplify notations we assume that $\mathbf{x}=\mathbf{y}$. 
Let $C$ be a component of $G$ with $\lambda^{\sf w}(C)\leq p$. Consider $\hat{F}=F[V(F)\cap V(C)]$ and let $\hat{\mathbf{y}}=\mathbf{y}\cap V(C)$. Note that $\hat{\mathbf{y}}$ may be empty.
Similarly, let $\hat{H}=H[V(H)\cap V(C)]$ and let $\hat{\mathbf{x}}=\mathbf{x}\cap V(C)$. It is straightforward to see that 
$C=(\hat{F},\hat{\mathbf{y}})\oplus_b(\hat{H},\hat{\mathbf{x}})$ and $G'$ has the component $C'=(\hat{F},\hat{\mathbf{y}})\oplus_b(\hat{H}',\hat{\mathbf{x}})$ where $(\hat{H}',\hat{\mathbf{x}})$ is obtained from $C$ by the cut reduction with respect to $p$. Since $\lambda^{\sf w}(G)\leq p$, we have that $\lambda^{\sf w}(C')=\lambda^{\sf w}(C)$ by Lemma~\ref{lem:cut-essential}.
Since all components of $G'$ are obtained from the components of $G$ in the described way, the claim follows.
\end{proof}

We need also some additional properties of the boundaried graphs obtained by cut reductions.

\begin{lemma}\label{lem:cut-additional}
Let $(H,\mathbf{x})$ be a properly $r$-boundaried weighted graph, 
and let $p\in\mathbb{N}\cup\{+\infty\}$. 
Then a properly $r$-boundaried weighted graph $H'$ obtained from $H$ by the cut reduction with respect to $p$ can be constructed in time $2^{\Oh(r)}\cdot n^{\Oh(1)}$ and it holds that 
\begin{itemize}
\item[(i)] $|V(H')|\leq 2^{2^{r-1}}+r$,
\item[(ii)] if for each component $C$ of $H$ that has at least one vertex outside $\mathbf{x}$, there is $v\in V(C)\setminus \mathbf{x}$ with $\lambda^{\sf w}(\mathbf{x},v)\leq p$, then
the same property holds for each component $C$ of $H'$ that has at least one vertex outside $\mathbf{x}$.
\end{itemize}
\end{lemma}

\begin{proof}
Assume that $H$ has $s$ components $H_1,\ldots,H_s$, $\mathbf{x}^i=\mathbf{x}\cap V(H_i)$ and $r_i=|\mathbf{x}^i|$ for $i\in\{1,\ldots,s\}$. 

To show that the cut reduction with respect to $p$ can be done in time $2^{\Oh(r)}\cdot n^{\Oh(1)}$, observe that for each $i\in\{1,\ldots,s\}$, we consider $2^{r_i}$ distinct partitions of $\mathbf{x}^i$ including the partition $(\mathbf{x}^i,\emptyset)$ and for each partition, find a separator of minimum weight ether for two parts of the partition or for $\mathbf{x}^i$ and a vertex in $V(H_i)\setminus\mathbf{x}^i$. Then we compute the partition $Z$ of $V(H)$ using  (\ref{eq:prod}) and this can be done in  time $2^{\Oh(r_i)}\cdot n^{\Oh(1)}$.
Since a separator of minimum weight can be found in polynomial time by the classical maximum flow/minimum cut algorithms (for example, by  the Dinitz's algorithm~\cite{Dinitz70,Dinitz06}), the running time is $(\sum_{i=1}2^{r_i-1})\cdot n^{\Oh(1)}$. Clearly, this can be written as   $2^{\Oh(r)}\cdot n^{\Oh(1)}$.

To show (i), let $i\in\{1,\ldots,s\}$. Recall that $\mathbf{x}^i$ has $2^{r_i-1}$ distinct partitions including the partition $(\mathbf{x}^i,\emptyset)$.
 By applying Observation~\ref{obs:prod} for the right part of (\ref{eq:prod}), we conclude that $H_i'$ has at most $2^{2^{r_i-1}}+r_i$ vertices. Then 
$H_i'$ has at most
$$\sum_{i=1}^s(2^{2^{r_i-1}}+r_i)\leq 2^{2^{r-1}}+r$$
vertices.

To prove (ii), assume that for some $i\in\{1,\ldots,s\}$, there is $v\in V(H_i)\setminus \mathbf{x}^i$ such that $\lambda^w_{H_i}(\mathbf{x}^i,v)\leq p$. Recall that by the definition of $H_i'$, there is $u\in V(H_i')\setminus \mathbf{x}^i$ with  $\lambda^w_{H_i'}(\mathbf{x}^i,u)\leq\lambda^w_{H_i}(\mathbf{x}^i,v)\leq p $ and the claim follows.
\end{proof}

For positive integers $r$ and  $s$, we define $\mathcal{H}_{r,s}$ as the family of all pairwise nonisomorphic properly $r$-boundaried weighted graphs $(G,\mathbf{x})$ with at most  $2^{2^{r-1}}+r$ vertices where the weights of edges are in $\{1,\ldots,s\}$ and for every component $C$ of $G$ with $V(C)\setminus\mathbf{x}\neq\emptyset$, there is a vertex $v\in V(C)\setminus \mathbf{x}$ such that 
$\lambda^{\sf w}(\mathbf{x},v)\leq s$. We also formally define $\mathcal{H}_{0,s}$ as the set containing the empty graph.

We need the following %straightforward 
lemma.

\begin{lemma}\label{lem:size}
For positive integers $r$ and $s$, 
%$|\mathcal{H}_{r,s}|\leq 2^{^{\binom{2^{2^{r-1}}+p}{2}}}s^{^{\binom{2^{2^{r-1}}+r}{2}}}$ 
$|\mathcal{H}_{r,s}|\leq (s+1)^{{\binom{2^{2^{r-1}}+r}{2}}}$
and  
$\mathcal{H}_{r,s}$ can be constructed in time $2^{2^{2^{\Oh(r)}}\log s}$.
\end{lemma}

\begin{proof}
Every graph in $\mathcal{H}_{r,s}$ has at most $2^{2^{r-1}}+r$ vertices. To construct $\mathcal{H}_{r,s}$, we first construct $2^{2^{r-1}}+r$ vertices and select $r$ vertices that compose $\mathbf{x}$. Then we consider at most $\binom{2^{2^{r-1}}+r}{2}$ pairs of distinct vertices $\{x,y\}$ such that $x$ or $y$ is not in $\mathbf{x}$,
and for each pair $\{x,y\}$ consider $s+1$ possibilities: $x$ and $y$ are nonadjacent or $xy$ is an edge of weight $i$ for $i\in\{1,\ldots,s\}$. This way we construct at most  $(s+1)^{^{\binom{2^{2^{r-1}}+r}{2}}}$ weighted graphs. Then for each constructed graph, we delete the components that do not contain any vertex of $\mathbf{x}$ and check whether for every component $C$ the obtained graph with $V(C)\setminus\mathbf{x}\neq\emptyset$, there is a vertex $v\in V(C)\setminus \mathbf{x}$ such that 
$\lambda^{\sf w}(\mathbf{x},v)\leq s$. If it holds, we include the constructed graph in $\mathcal{H}_{r,s}$ unless it already contains an isomorphic  $r$-boundaried weighted graphs. It is straighforward to see that the running time of the procedure is  $2^{2^{2^{\Oh(r)}}\log s}$.
\end{proof}

\paragraph{Variants of \probWCGC.}
To apply the recursive understanding technique, we also have to solve a special variant of \probWCGC tailored for recursion. To define it, we first introduce the following  variant of the problem. The difference is that a solution should be chosen from a given  subset of  edges.

\defproblema{\probAWCGC}%
{A weighted graph $G$ with an edge weight function $w\colon E(G)\rightarrow \mathbb{N}$, $L\subseteq E(G)$,
 a $t$-tuple $\Lambda=\langle \lambda_{1},\ldots,\lambda_{t}\rangle$, where $\lambda_i\in\mathbb{N}\cup\{+\infty\}$ for $i\in\{1,\ldots,t\}$ and  $\lambda_1\leq\ldots\leq \lambda_t$, and a nonnegative integer $k$.}%
{Decide whether there is a set $F\subseteq L$ with ${\sf w}(F)\leq k$ such that $G-F$ has $t$ connected components $G_1,\ldots, G_t$ where each $\lambda^{\sf w}(G_i)\geq \lambda_i$ for $i\in\{1,\ldots,t\}$.}

Clearly, if $L=E(G)$, then \probAWCGC is \probWCGC.
Let $(G,w,L,\Lambda,k)$ be an instance of \probAWCGC.
We say that  $F\subseteq L$ with ${\sf w}(F)\leq k$ such that $G-F$ has $t$ connected components $G_1,\ldots, G_t$ where each $\lambda^{\sf w}(G_i)\geq \lambda_i$ for $i\in\{1,\ldots,t\}$ is a \emph{solution} for the instance. 

Now we define \probBAWCGC.

\defproblema{\probBAWCGC}%
{A weighted $r$-boundaried connected graph $(G,\mathbf{x})$ with an edge weight function $w\colon E(G)\rightarrow \mathbb{N}$, $L\subseteq E(G)$,
 a $t$-tuple $\Lambda=\langle \lambda_{1},\ldots,\lambda_{t}\rangle$, where $\lambda_i\in\mathbb{N}\cup\{+\infty\}$ for $i\in\{1,\ldots,t\}$ and  $\lambda_1\leq\ldots\leq \lambda_t$, and a nonnegative integer $k$ such that 
$r\leq 4k$ and $k\geq t-1$.
}%
{For each weighted properly $r$-boundaried graph $(H,\mathbf{y})\in \mathcal{H}_{r,2k}$ and each $\hat{\Lambda}=\langle \hat{\lambda}_1,\ldots,\hat{\lambda}_s\rangle\subseteq \Lambda$, 
find the minimum $0\leq \hat{k}\leq k$ such that $((G,\mathbf{x})\oplus_b(H,\mathbf{y}),w,L,\hat{\Lambda},\hat{k})$ is a yes-instance of \probAWCGC and output a solution $F$ for this instance or output $\emptyset$ if $\hat{k}$ does not exist.}

Slightly abusing notation, we use $w$ to denote the weights of edges of $G$ and $H$. Notice that  \probBAWCGC is neither decision nor optimization problem, and its solution is a list of subsets of $L$. Observe also that a solution of \probBAWCGC is not necessarily unique. Still, for any two solutions, that is, lists $\mathcal{L}_1$ and $\mathcal{L}_2$ of subsets of $L$, the following holds: for each  weighted properly $r$-boundaried graph $(H,\mathbf{y})\in \mathcal{H}_{r,2k}$ and each $\hat{\Lambda}=\langle \hat{\lambda}_1,\ldots,\hat{\lambda}_s\rangle\subseteq \Lambda$, the lists $\mathcal{L}_1$ and $\mathcal{L}_2$ contain the sets of the same weight. 
To solve \probAWCGC, it is sufficient to solve \probBAWCGC for $r=0$. If the output contains nonempty set for $\hat{\Lambda}=\Lambda$, we have a yes-instance of \probAWCGC. If the output contains empty set for this $\hat{\Lambda}$, we should verify additionally whether $\emptyset$ is a solution, that is, whether $\Lambda=\{\lambda_1\}$ and $\lambda^{\sf w}(G)\geq \lambda_1$.  To apply the recursive understanding technique, we first solve  \probBAWCGC for $(q,2k)$-unbreakable graphs for some appropriate value of $q$  and then use this result for the general case of \probBAWCGC.

\paragraph{Tools for randomized separation.}
To solve \probBAWCGC for $(q,2k)$-unbreakable graphs, we use the \emph{random separation} technique introduces by  Cai, Chan and Chan in~\cite{CaiCC06}. To avoid dealing with randomized algorithms, we use the following lemma stated by Chitnis et al. in~\cite{ChitnisCHPP16}.

\begin{lemma}[\cite{ChitnisCHPP16}]\label{lem:derand}
Given a set $U$ of size $n$ and integers $0\leq a,b\leq n$, one can construct in time $2^{\Oh(\min\{a,b\}\log (a+b))}\cdot n\log n$ a family $\mathcal{S}$ of at most  $2^{\Oh(\min\{a,b\}\log (a+b))}\cdot \log n$ subsets of $U$ such that the following holds: for any sets $A,B\subseteq U$, $A\cap B=\emptyset$, $|A|\leq a$, $|B|\leq b$, there exists a set $S\in \mathcal{S}$ with $A\subseteq S$ and $B\cap S=\emptyset$. 
\end{lemma}

We also use the following results of Fomin and Villanger~\cite{FominV12}.

\begin{lemma}[\cite{FominV12}]\label{lem:enum-sets}
Let $G$ be a graph. For every $v\in V(G)$ and nonnegative integers $b$ and $f$, the number of connected subsets $B\subseteq V(G)$ such that 
\begin{itemize}
\item[(i)] $v\in B$,
\item[(ii)] $|B|=b+1$, and
\item[(iii)] $|N_G(B)|=f$
\end{itemize}
is at most $\binom{b+f}{b}$. Moreover, all such sets $B$ could be enumerated in time
$\Oh(\binom{b+f}{f}\cdot b(b+f)\cdot n)$.
\end{lemma}

We conclude this section with a combinatorial observation.

Let $G$ be a graph, $u\in V(G)$, and let $r$ be a positive integer. We construct the subgraph $B_r(u)$ using a modified breadth-first search algorithm. Recall that in the standard breadth-first search algorithm (see, e.g.,~\cite{CormenLRS09}) starting from $u$, we first label $u$ by $\ell(u)=0$ and put $u$ into a queue $Q$. Then we iterate as follows: if $Q$ is nonempty, then take the first vertex $x$ in the queue and for every nonlabeled neighbor $y$, assign $\ell(y)=\ell(x)+1$ and put $y$ into $Q$. We start in the same way by assigning $u$ the label $\ell(u)=0$ and putting $u$ into $Q$. Then while $Q$ is nonempty and the first element $x$ has the label $\ell(x)\leq r-1$, we consider  arbitrary chosen $\min\{r,d_G(x)\}$ vertices $y\in N_G(x)$, assign to unlabeled vertices $y$ the label $\ell(y)=\ell(x)+1$ and put them into $Q$. The graph $B_r(u)$ is the subgraph of $G$ induced by the labeled vertices $v$ with $\ell(v)\leq r$. We say that $B_r(x)$ is an \emph{$r$-restricted BFS subgraph} of $G$. 
Note that such a subgraph is not unique. We need the following properties of $B_r(u)$.

\begin{lemma}\label{lem:bfs}
Let $G$ be a weighted graph, $u\in V(G)$, and let $r$ be a positive integer. 
Then for an $r$-restricted BFS subgraph $B_r(u)$ of $G$, it holds the following.
\begin{itemize} 
\item[(i)] If $G$ is weighted $k$-connected and $|V(G)|\geq r-k+1$, then for every connected set $X\subseteq V(G)$ such that $u\in X$ and $|X|\leq r-k$, ${\sf w}(X\cap V(B_r(u)),V(B_r(u))\setminus X)\geq k$.
\item[(ii)] For $r\geq 2$, $|V(B_r(x))|\leq (r^{r+1}-1)/(r-1)=2^{\Oh(r\log r)}$.  
\end{itemize}
\end{lemma}

\begin{proof}
To show (i), assume that $G$ is weight $k$-connected and $|V(G)|\geq r-k+1$. Let also $X\subseteq V(G)$ be a connected set such that $u\in X$ and $|X|\leq r-k$. Assume that the vertices of $B_r(u)$ have the labels $\ell(v)\leq r$ assigned by our modified BFS algorithm. 

Because $u\in X$, we have that $u\in X\cap V(B_r(u))\neq\emptyset$. Suppose that $X\cap V(B_r(u))$ contains a vertex $v$ such that $N_G(v)\setminus V(B_r(u))\neq\emptyset$. Let $v$ be such a vertex with the minimum value of $\ell(v)$. If $\ell(v)\leq r-1$, we have that $d_G(v)>r$ as otherwise all the vertices of $N_G(v)$ would be labeled. Therefore, $v$ has at least $r$ neighbors in $B_r(u)$. Since $|X|\leq r-k$,  at most $r-k$ of the neighbors of $v$ are in $X$. Hence, there are at least $k$ neighbors of $v$ in $B_r(u)$ that are not included in $X$. We obtain that 
 ${\sf w}(X\cap V(B_r(u)),V(B_r(u))\setminus X)\geq {\sf w}(v,N_{B_r(u)}(v)\setminus X)\geq |E(v,N_{B_r(u)}(v)\setminus X)| \geq k$.
Suppose that $\ell(v)=r$. Recall that $X$ is connected, $u\in X$ and $|X|\leq r-k$. Hence, there is a $(u,v)$-path $P$ of length at most $r-1$. Since $v\in X\cap V(B_r(u))$ is a vertex with the minimum value of $\ell(v)$ that has an unlabeled neighbor, there are adjacent $x,y\in V(P)$ such that $x,y\in X\cap V(B_r(u))$ and $\ell(x)\leq \ell(y)-2$. Note that $y$ was not labeled when the algorithm constructing $B_r(u)$ considered $x$. This implies that $d_G(x)>r$. Then by the same arguments as above, we obtain that 
 ${\sf w}(X\cap V(B_r(u)),V(B_r(u))\setminus X)\geq {\sf w}(x,N_{B_r(u)}(x)\setminus X)\geq |E(x,N_{B_r(u)}(x)\setminus X)| \geq k$.

Suppose now that $N_G[X]\subseteq V(B_r(u))$. Then $E(X,V(G)\setminus X)=E(X\cap V(B_r(u)),V(B_r(u))\setminus X)$. Because $G$ is weight $k$-connected, we have that
 ${\sf w}(X\cap V(B_r(u)),V(B_r(u))\setminus X)={\sf w}(X,V(G)\setminus X)\geq k$.

The second claim of the lemma follows immediately from the definition of $B_r(u)$.
\end{proof}

\subsection{High connectivity phase}\label{sec:high}
In this section we construct an algorithm for \probBAWCGC for connected $(q,2k)$-unbreakable graphs. As the first step, we solve  \probAWCGC.

\begin{lemma}\label{lem:annot}
\probAWCGC can be solved and a solution can be found (if exists) in time $2^{\Oh(q(q+k)\log(q+k))}\cdot n^{\Oh(1)}$ for connected $(q,2k)$-unbreakable graphs. 
\end{lemma}

\begin{proof}
Let $(G,w,L,\Lambda,k)$ be an instance of \probAWCGC where $G$ is a connected $(q,2k)$-unbreakable graph. Let also $\Lambda=\langle \lambda_1,\ldots,\lambda_t\rangle$, $\lambda_1\leq \ldots\leq\lambda_t$.

Clearly, the problem is easy if $t=1$ as it sufficient to check whether $\lambda^{\sf w}(G)\geq \lambda_1$ in polynomial time using, e.g., the algorithm of Stoer and Wagner~\cite{StoerW97}.
Also the problem is trivial if $t>k+1$: 
 $(G,w,L,\Lambda,k)$ is a no-instance, because the connected graph $G$ can be separated into  at most $k+1$ components by at most $k$ edge deletions.  Hence, from now we assume that $2\leq t\leq k+1$. 
 
If $|V(G)|\leq 3q$, we solve \probAWCGC using brute force. We consider all the 
possibilities to select a set of edges $F\subseteq L$ of the total weight at most $k$, and for each choice, we check whether $G-F$ has  $t$ components $G_1,\ldots,G_t$ with $\lambda^{\sf w}(G_i)\geq \lambda_i$ for $i\in\{1,\ldots,t\}$. Note that we can check in polynomial time whether the components of $G-F$ satisfy the connectivity constraints. We compute the weighted connectivities of the components using the minimum cut algorithm of Stoer and Wagner~\cite{StoerW97}. Then we sort the components by their weighted connectivities and assume that $\lambda^{\sf w}(G_1)\leq\ldots\leq\lambda^{\sf w}(G_t)$. It remains to check whether $\lambda^{\sf w}(G_i)\geq \lambda_i$ for $i\in\{1,\ldots,t\}$.
The total running time in this case is $2^{\Oh(k\log q)}$. From now we assume that $|V(G)|>3q$.

Suppose that $(G,w,L,\Lambda,k)$ is a yes-instance of \probAWCGC and let $F\subseteq L$ be a solution. Let $G_1,\ldots,G_t$ be the components of $G-F$ and  $\lambda^{\sf w}(G_i)\geq \lambda_i$ for $i\in\{1,\ldots,t\}$. 

Notice that there is a component with at least $q+1$ vertices. Otherwise, we can partition the components into two families $\{G_i\mid i\in I\}$ and $\{G_i\mid i\in \{1,\ldots,t\}\setminus I\}$ for some $I\subset \{1,\ldots,t\}$ in such a way that $|\sum_{i\in I}|V(G_i)|-\sum_{i\in\{1,\ldots,t\}\setminus I}|V(G_i)||\leq q$. Then for $A=\bigcup_{i\in I}V(G_i)$ and $B=\bigcup_{i\in\{1,\ldots,t\}\setminus I}V(G_i)$, we have that $|A|>q$, $|B|>q$ and $E(A,B)\subseteq F$. Since ${\sf w}(F)\leq k$, we obtain that ${\sf w}(A,B)\leq k$ but this contradicts the condition that $G$ is $(q,2k)$-unbreakable. Hence, we can assume that $|V(G_i)|\geq q+1$ for some $i\in\{1,\ldots,t\}$. Observe that the total number of vertices in the other components is at most $q$. Otherwise, 
for    $A=V(G_i)$ and $B=\bigcup_{j\in\{1,\ldots,t\}\setminus \{i\}}V(G_j)$, we have that $|A|>q$, $|B|>q$ and $E(A,B)\subseteq F$, and this contradicts the unbreakability of $G$. 
We say that $G_i$ is a \emph{big} components and call the other components \emph{small}.

For each $i\in\{1,\ldots,t\}$, we verify whether there is a solution $F$ where $\lambda_i$ is the connectivity constraint for the big component of $G-F$. 

Assume that $\lambda_i>k$.

Let $\alpha$ be a positive integer. We say that two vertices $u,v\in V(G)$ are \emph{$\alpha$-equivalent} if $\lambda^{\sf w}(u,v)\geq \alpha$.  It is straightforward to verify that this is an equivalence relation on $V(G)$. We call the classes equivalence of $V(G)$ the \emph{$\alpha$-classes}. 
Let $X_1,\ldots,X_s$ be the $(k+1)$-classes. We claim that if  $\lambda^{\sf w}(G_j)\geq k+1$ for a component $G_i$ of $G-F$, then $G_i=G[X_h]$ for $h\in\{1,\ldots,s\}$. To see this, notice that for each $h\in\{1,\ldots,s\}$, $X_h\subseteq V(G_j)$ for $j\in\{1,\ldots,t\}$, because any two vertices of $X_h$ are in the same component of $G-F$ since ${\sf w}(F)\leq k$. From the other side, if two vertices $u,v$ are in distinct $(k+1)$-classes, then it cannot happen that both $u,v\in V(G_i)$ because $G_i$ is weight $(k+1)$-connected.   

We use these observations to check whether there is a solution $F$ for $(G,w,L,\Lambda,k)$ such that for the big component $G_i$ of $G-F$, it holds that $\lambda^{\sf w}(G_i)\geq \lambda_i>k$.

We find the $(k+1)$-classes $X_1,\ldots,X_s$ in polynomial time using the flow algorithms~\cite{Dinitz70,Dinitz06}. Then we find $j\in\{1,\ldots,s\}$ such that $|X_j|\geq q+1$ and $\lambda^{\sf w}(G[X_j])\geq \lambda_i$. If such a set $X_j$ does not exist, we conclude that $(G,w,L,\Lambda,k)$ is a no-instance and stop, because  there is no big component $G_i$ with $\lambda^{\sf w}(G_i)\geq \lambda_i$.
We also return {\sf NO} and stop  if $|V(G)|-|X_j|>q$, because  the total number of vertices in small components for any solution $F$ is at most $q$. 
If $G_i=G[X_j]$ is the big component for the solution $F$, then $E(X_j,V(G)\setminus X_j)\subseteq L$ and ${\sf w}(X_j,V(G)\setminus X_j)\leq k$. We verify these conditions and return {\sf NO} and stop if one of them is violated. Otherwise, we use brute force and consider all possible choices a set of edges $F'\subseteq L$ of $G'=G-X_j$ with ${\sf w}(F')\leq k-{\sf w}(X_j,V(G)\setminus X_j)$, and for each selection, we check whether  $F=E(X_j,V(G)\setminus X_j)\cup F'$ is a solution for $(G,w,L,\Lambda,k)$. Since $|V(G')|\leq q$, this could be done in time $2^{\Oh(k\log q)}\cdot n^{\Oh(1)}$.

From now we assume that $\lambda_i\leq k$. %Let $\mu=\lambda_i-1$.

Assume again that $(G,w,L,\Lambda,k)$ is a yes-instance of \probAWCGC, $F\subseteq L$ is a solution, and $G_1,\ldots,G_t$ are the components of $G-F$ where $G_i$ is the big component.
Let $A=\bigcup_{j\in\{1,\ldots,t\}\setminus \{i\}}V(G_j)$. Recall that $|A|\leq q$. Let also $X\subseteq V(G_i)$ be the set of vertices of $G_i$ that have neighbors in $A$. Note that $|X|\leq k$. For each $u\in X$, we consider 
a $(q+\lambda_i)$-restricted BFS subgraph  $B(u)=B_{q+\lambda_i}(u)$ of $G_i$.  Let $B=\bigcup_{u\in X}V(B(u))$. By Lemma~\ref{lem:bfs}, we have that $|V(B(u))|=2^{\Oh((q+k)\log(q+k))}$ since  $\lambda_i\leq k$. Hence, $|B|=2^{\Oh((q+k)\log(q+k))}$. Note also that $|B|\geq q+1$. We say that a set $S\subseteq V(G)$ is \emph{$(A,B)$-good} or, simply, \emph{good} if $B\subseteq S$ and $A\cap S=\emptyset$.
By Lemma~\ref{lem:derand}, we can construct in time $2^{\Oh(q(q+k)\log(q+k))}\cdot n\log n$ a family $\mathcal{S}$ of at most  $2^{\Oh(q(q+k)\log(q+k))}\cdot \log n$ subsets of $V(G)$ such that if $(G,w,L,\Lambda,k)$ is a yes-instance and $A$ and $B$ exist for some solution, then $\mathcal{S}$ contains an $(A,B)$-good set.

We construct such a family $\mathcal{S}$, and for each $S\in \mathcal{S}$, we look for  $F\subseteq L$ such that the following holds:
\begin{itemize}
\item[(i)] ${\sf w}(F)\leq k$,
\item[(ii)] $G-F$ has $t$ components  $G_1,\ldots,G_t$ such that each $G_j$ is weight $\lambda_j$-connected and $|V(G_i)|>q$, and
\item[(iii)] $S\subseteq V(G_i)$.% and $V(G_j)\cap S=\emptyset$ for $j\in \{1,\ldots,t\}$, $j\neq i$.
\end{itemize}
We describe the algorithm that produces the {\sf YES} answer if $S$ is good and, moreover, if the algorithm outputs {\sf YES}, then  $(G,w,L,\Lambda,k)$ is a yes-instance of \probAWCGC. Note that the algorithm can output the false {\sf NO} answer if $S$ is not a good set. Nevertheless, because for an yes-instance of \probAWCGC, we always have a good set $S\in\mathcal{S}$, we have that  
$(G,w,L,\Lambda,k)$ is a yes-instance if and only if the algorithm outputs {\sf YES} for at least one $S\in\mathcal{S}$. 

The algorithm uses the following property of $(A,B)$-good sets.

\medskip
\noindent
{\bf Claim~A.} {\it
If $S$ is an $(A,B)$-good set, then for each component $H$ of $G-S$, either $V(H)\subseteq V(G_i)$ or $V(H)\cap V(G_i)=\emptyset$.
}

\begin{subproof}[Proof of Claim~A]
To show it, assume for the sake of contradiction that $H$ contains both vertices of the big component $G_i$ and some small component. Then there are two adjacent vertices $x,y\in V(H)$ such that $x\in V(G_i)$ and $y\in V(G_j)$ for a small component $G_j$. Then $x\in V(B(x))\subseteq B$ contradicting that $B\subseteq S$.   
\end{subproof}

We apply a number of reduction rules that either increase the set $S$ or conclude that $S$ is not good and stop. For each rule, we show that it is \emph{safe} in the sense that if we increase $S$, then if the original $S$ was good, then the obtained set is good as well, and if we conclude that the original $S$ is not good, then this is a correct conclusion and, therefore, we can return {\sf NO} and stop. We underline that whenever we return {\sf NO} in the rules, this means that we discard the current choice of $S$.

\begin{reduction}\label{red:small-s}
If $|S|\leq q$, then return {\sf NO} and stop.
\end{reduction}

To see that the rule is safe, it is sufficient to note that if $S$ is $(A,B)$-good, then $|S|\geq |B|\geq q+1$.

Denote by $H_1,\ldots,H_s$ the components of $G-S$. By Claim~A, either $V(H_j)\subseteq G_i$ or $V(H_j)\cap G_i=\emptyset$ for $j\in\{1,\ldots,s\}$ if $S$ is good. 

\begin{reduction}\label{red:sort}
For every $j\in\{1,\ldots,s\}$, if  $E(V(H_j),S)\setminus L\neq \emptyset$ or ${\sf w}(V(H_j),S)\geq k+1$ or $|V(H_j)|>q$, then set $S=S\cup V(H_j)$.
\end{reduction}

To show safeness, assume that $S$ is $(A,B)$-good. Observe that if $E(V(H_j),S)\setminus L\neq \emptyset$, then any $F\subseteq L$ does not separate $S$ and $V(H_j)$. 
Similarly, if ${\sf w}(V(H_j),S)\geq k+1$, then a set $F\subseteq L$ with ${\sf w}(F)\leq k$ cannot separate $S$ and $V(H_j)$. By Claim~A, we conclude that  $V(H_j)\subseteq G_i$. 
If $|V(H_j)|>q$, then $V(H_j)\subseteq V(G_i)$ by Claim~A because the total number of vertices in the small components is at most $q$. 
Since each $H_j$ has no vertex adjacent to a vertex of any small component $G_h$, we obtain that if 
$S=S\cup V(H_j)$ is $(A,B)$-good.

To simplify notations, we assume that $H_1,\ldots,H_s$ are the components of $G-S$ for the (possibly) modified by Reduction Rule~\ref{red:sort} set $S$. 
%Let $U=\{u\in S\mid N_G(u)\cap V(H_j)\neq\emptyset\text{ for some }j\in\{1,\ldots,s\}\}$. %We exhaustively apply the following rule.  

\begin{reduction}\label{red:conn}
If there is $u\in S$ adjacent to a vertex of $H_j$ for some $j\in\{1,\ldots,s\}$ and there is a  connected set $Z\subseteq S$ such that a) $u\in Z$, b) $|Z|\leq q$, c) ${\sf w}(Z,S\setminus Z)\leq \lambda_i-1$, then set $S=S\cup V(H_j)$.  
\end{reduction}

Observe that by Lemma~\ref{lem:enum-sets}, for each $u\in S$, we can list all the sets $Z$ satisfying a)--c) in time $2^{\Oh(k\log(q+k))}\cdot n$ because $\lambda_i\leq k$.  

To prove that the rule is safe, assume that $S$ is an $(A,B)$-good, and there is $u\in S$ adjacent to a vertex of $H_j$ for some $j\in\{1,\ldots,s\}$ and there is a  connected set $Z\subseteq S$ such that a)--c) hold. By Claim~A, either $V(H_j)\subseteq V(G_i)$ or $V(H_j)\cap G_i=\emptyset$. Suppose that $V(H_j)\cap G_i=\emptyset$. Since $S$ is good, we have that $V(B(u))\subseteq S$. But then by Lemma~\ref{lem:bfs}, ${\sf w}(Z\cap V(B(u)),V(B(u))\setminus Z)\geq \lambda_i$ for any connected set satisfying a) and b); a contradiction. Hence, $V(H_j)\subseteq V(G_i)$. 
Because $H_j$ has no vertex adjacent to a vertex of any small component $G_h$, we obtain that  $S=S\cup V(H_j)$ is $(A,B)$-good.

We apply Reduction Rule~\ref{red:conn} exhaustively recomputing the components of $G-S$ after each modification of $S$. 

\begin{reduction}\label{red:stop}
If there is  a  connected set $Z\subseteq S$ such that $|Z|\leq q$ and ${\sf w}(Z,V(G)\setminus Z)\leq \lambda_i-1$, then set return {\sf NO} and stop.  
\end{reduction}

To see that the rule is safe, observe that because Reduction Rule~\ref{red:conn} cannot be applied further, if there is  a  connected set $Z\subseteq S$ such that $|Z|\leq q$, then 
$N_G(Z)\subseteq S$. Then
it is sufficient to observe that if $S$ is $(A,B)$-good, then $S\subseteq V(G_i)$ and since $G_i$ is weight $\lambda_i$-connected, ${\sf w}(Z,V(G)\setminus Z)\geq \lambda_i$ for every $Z\subseteq S$. Observe also that the connected sets $Z\subseteq S$ such that $|Z|\leq q$ and ${\sf w}(Z,V(G)\setminus Z)\leq \lambda_i-1$
 can be enumerated in time $2^{\Oh(k\log(q+k))}\cdot n^2$ by Lemma~\ref{lem:enum-sets}.

Assume that we do not stop while executing Reduction Rule~\ref{red:stop}. Let $S$ be the set constructed from the original set $S$ by the rules that were applied so far. 
%Let also $H_1,\ldots,H_s$ be the components of $G-S$. 

\medskip
\noindent
{\bf Claim~B.} {\it
If $S$ is an $(A,B)$-good set, then if for a component $H$ of $G-S$, there is $v\in V(H)$ such that $\lambda^{\sf w}(v,S)<\lambda_i$, then $V(H)\cap V(G_i)=\emptyset$.
}

\begin{subproof}[Proof of Claim~B]
By Claim~A, we have that either $V(H)\subseteq V(G_i)$ or $V(H)\cap V(G_i)=\emptyset$. Since the inclusion $V(H)\subseteq V(G_i)$ would contradict the weight $\lambda_i$-connectedness of $G_i$, we conclude that $V(H)\cap V(G_i)=\emptyset$.
\end{subproof}

Let  $H_1,\ldots,H_s$ be the components of $G-S$. We set 
$$I=\{j\in\{1,\ldots,s\}\mid \text{there is }v\in V(H_j)\text{ such that }{\sf w}(v,S)<\lambda_i\}.$$

\begin{reduction}\label{red:stop-two}
If $|\bigcup_{j\in I}V(H_j)|>q$ or ${\sf w}(\bigcup_{j\in I}V(H_j),S)>k$, then return {\sf NO} and stop.  
\end{reduction}

To see that the rule is safe, observe that if $S$ is $(A,B)$-good, then by Claim~B, the vertices of $\bigcup_{j\in I}V(H_j)$ should be in small components. Since the number of vertices in the small components cannot exceed $q$ and ${\sf w}(\bigcup_{j\in I}V(H_j),S)\leq k$, we can return {\sf NO} and stop if at least one of these condition is not fulfilled.

\medskip
\noindent
{\bf Claim~C.} {\it
For any $J\subseteq\{1,\ldots,s\}$ such that $I\subseteq J$ and ${\sf w}(\bigcup_{j\in J}V(H_j),S)\leq k$, the graph $G'=G-\bigcup_{j\in J}V(H_j)$ is weight $\lambda_i$-connected.
}

\begin{subproof}[Proof of Claim~C]
To obtain a contradiction, assume that $\lambda^{\sf w}(G')<\lambda_i$. Then there is a partition $(X,Y)$ of $V(G')$ such that ${\sf w}(X,Y)<\lambda_i$. Since $S\subseteq V(G')$, we assume without loss of generality that $X\cap S\neq\emptyset$.

Suppose that $Y\cap S=\emptyset$. Then $S\subseteq X$. This means that there is $v\in Y$ such that $\lambda^{\sf w}(v,S)\leq\lambda_i-1$ and $v\in V(H_j)$ for some $j\in\{1,\ldots,k\}\setminus I$. But by the definition of the set of indices $I$, we should have $j\in I$; a  contradiction. Therefore, $Y\cap S\neq\emptyset$. We assume without loss of generality that $|X\cap S|\leq|Y\cap S|$. 

Suppose that $|X\cap S|\leq q$. If there is a vertex $v\in V(H_j)$ for some $j\in\{1,\ldots,s\}$ such that $v\in N_G(X\cap S)$, then we could apply Reduction Rule~\ref{red:conn} for $Z\subseteq S$ that is the set of vertices of the component of $G[X\cap S]$ containing a neighbor of $v$, and we would include $V(H_j)$ into $S$. Hence, $N_G[X\cap S]\subseteq S$. But then we could apply Reduction Rule~\ref{red:stop} for the set of vertices $Z$ of one of the components $G[X\cap S]$, and we would stop; a contradiction. This means that  $q<|X\cap S|\leq |Y\cap S|$.  

Let $R=E(S,\bigcup_{j\in J}V(H_j))\cup E(X,Y)$. Since ${\sf w}(\bigcup_{j\in J}V(H_j),S)\leq k$ and ${\sf w}(X,Y)<\lambda_i$, we have that ${\sf w}(R)\leq 2k$. Observe that $R$ separates $X\cap S$ and $Y\cap S$ in $G$ but this contradicts the $(q,2k)$-unbreakability of $G$. We conclude that there is no partition $(X,Y)$ of $V(G')$ such that ${\sf w}(X,Y)<\lambda_i$ and, therefore, $\lambda^{\sf w}(G')\geq \lambda_i$.
\end{subproof}

By applying Reduction Rules~\ref{red:small-s}--\ref{red:stop-two}, we either increase $S$ or stop. Now we have to find an $F\subseteq L$ such that (i)--(iii) are fulfilled and, by applying Claims~A and B, we impose two additional constraints:
\begin{itemize}
\item[(iv)] for every $j\in \{1,\ldots,s\}\setminus I$, either $V(H_j)\subseteq V(G_i)$ or $V(H_j)\cap V(G_i)=\emptyset$, 
\item[(v)] for every $j\in I$,  $V(H_j)\cap V(G_i)=\emptyset$.
\end{itemize}
Note that by Claim~C, we automatically obtain that $\lambda^{\sf w}(G_i)\geq\lambda_i$ if (i), (iii)-(v) are fulfilled. Also because of Reduction Rule~\ref{red:small-s}, we have that $|V(G_i)|>q$ if (iii) holds.
Hence, we can replace (ii) by the relaxed condition:
\begin{itemize}
\item[(ii)] $G-F$ has $t$ components  $G_1,\ldots,G_t$ such that $G_j$ is weight $\lambda_j$-connected for $j\in\{1,\ldots,t\}$, $j\neq i$.
\end{itemize}
We find $F$, if such a set exists, by  a dynamic programming algorithm. 

Let $Q_0=G[\bigcup_{j\in I}V(H_j)]$. Note that it can happen that $I=\emptyset$ and $Q_0$ is empty in this case. Let $Q_1,\ldots,Q_r$ be the remaining components $H_j$ for $j\in\{1,\ldots,k\}\setminus I$.  

For every $j\in\{0,\ldots,r\}$, we define the function $f_j(\Lambda',\ell)$ with the values in $2^L\cup\{false\}$ where $\Lambda'\subset \Lambda$ excluding $\lambda_i$ and $\ell\leq k$ is a nonnegative integer.
We set 
$$f_0(\emptyset,\ell)=
\begin{cases}
false &\mbox{if }I\neq\emptyset,\\
\emptyset & \mbox{if }I=\emptyset,
\end{cases}
$$
and $f_j(\emptyset,\ell)=\emptyset$ for $j\in\{1,\ldots,r\}$. For nonempty $\lambda'=\langle\lambda_1',\ldots,\lambda_p'\rangle$ and $j\in\{0,\ldots,r\}$, we set $f_j(\Lambda',\ell)=F\subseteq L$ if there is $F\subseteq L$ such that
\begin{itemize}
\item[a)] ${\sf w}(F)\leq \ell$,
\item[b)] $E(S,V(Q_j))\subseteq F$ and $F\setminus E(S,V(Q_j))\subseteq E(Q_j)$,
\item[c)] $Q_j-(F\setminus E(S,V(Q_j)))$ has $p$ components $C_1,\ldots,C_p$ such that $\lambda^{\sf w}(C_h)\geq\lambda_h$ for $h\in\{1,\ldots,p\}$,
\end{itemize}
and we set $f_j(\Lambda',\ell)=false$ otherwise. Note that $f_j(\Lambda',\ell)$ is not uniquely defined as we may  find distinct $F$ satisfying a)--c). Respectively, we pick one of the possible values to define $f_j(\Lambda',\ell)$.

For each nonempty $\Lambda'$ and $\ell$, we compute the value of each $f_j(\Lambda',\ell)$ by brute force by considering subsets of $E(Q_j)$ of weight at most $\ell-{\sf w}(S,V(Q_j))$ using the fact that $|V(Q_j)|\leq q$ because of Reduction Rules~\ref{red:sort} and \ref{red:stop-two}. Clearly, it can be done in time $2^{\Oh(k\log q)}$. Since there are at most $2^t(k+1)$ pairs $(\Lambda',\ell)$ and $t\leq k+1$, the table of all the  values of $f_j(\Lambda',\ell)$ can be constructed in time $2^{\Oh(k\log q)}$.
Since $r\leq n$, all the tables could be constructed in time $2^{\Oh(k\log q)}\cdot n$.

For every $j\in\{0,\ldots,r\}$, we define the function $g_j(\Lambda',\ell)$ with the values in $2^L\cup\{false\}$ where $\Lambda'\subset \Lambda$ excluding $\lambda_i$ and $\ell\leq k$ is a nonnegative integer. We set $g_0(\Lambda',\ell)=f_0(\Lambda',\ell)$ and for $j\in\{1,\ldots,p\}$,
we set $g_j(\Lambda',\ell)=F\subseteq 2^L$ if there are  $\Lambda^1,\Lambda^2$ such  that $\Lambda'=\Lambda^1+\Lambda^2$ and there are nonnegative integers $\ell_1,\ell_2$ such that $\ell_1+\ell_2=\ell$ such that 
\begin{itemize}
\item[d)] $g_{j-1}(\Lambda^1,\ell_1)\neq false$ and $f_{j}(\Lambda^2,\ell_2)\neq false$, and
\item[e)] $F=g_{j-1}(\Lambda^1,\ell_1)\cup f_{j}(\Lambda^2,\ell_2)$,
\end{itemize}
and $g_j(\Lambda',\ell)=false$ otherwise.
As with $f_j(\Lambda',\ell)$, we have that $g_j(\Lambda',\ell)$ is not uniquely defined because d) and e) could hold for distinct  $\Lambda^1,\Lambda^2$ and/or $\ell_1,\ell_2$, and as with 
$f_j(\Lambda',\ell)$, we pick one feasible $F$.

Observe that for each $j\in\{1,\ldots,r\}$, the table of values of $g_j(\Lambda',\ell)$ has at most $2^r(k+1)$ entries, and can be constructed from the table  of value $g_{j-1}(\Lambda',\ell)$ 
in time $2^{\Oh(k\log q)}$ using the already computed table of values of  $f_j(\Lambda',\ell)$. We conclude that all the tables of values of $g_j(\Lambda',\ell)$ can be constructed in time 
$2^{\Oh(k\log q)}\cdot n$.

By the definition of $f_j(\Lambda',\ell)$ and $g_j(\Lambda',\ell)$, we obtain the following straightforward claim.

\medskip
\noindent
{\bf Claim~D.} {\it
For every $\lambda'=\langle\lambda_1',\ldots,\lambda_p'\rangle\subset \Lambda$ excluding $\lambda_i$, nonnegative integer $\ell\leq k$ and $j\in\{0,\ldots,r\}$,
$g_j(\Lambda',\ell)=F\neq false$ if and only if 
\begin{itemize}
\item $F\subseteq L$ and ${\sf w}(F)\leq \ell$,
\item $G-F$ has $p+1$ components $C_1,\ldots,C_p,G'$ such that
\begin{itemize}
\item $S\subseteq V(G')$,
\item $V(Q_0)\cap V(G')=\emptyset$,
\item for each $h\in \{1,\ldots,j\}$, either $V(Q_h)\subseteq V(G')$ or $V(Q_h)\cap V(G')=\emptyset$,
\item for each  $h\in\{j+1,\ldots,r\}$, $V(Q_h)\subseteq V(G')$,
\item for each $h\in\{1,\ldots,p\}$, $\lambda^{\sf w}(C_h)\geq\lambda_h'$.
\end{itemize}
\end{itemize}
}

By Claim~D, we have that   $F\subseteq L$ satisfying (i)--(v) exists if and only if $g_r(\Lambda,k)\neq false$, and if $g_r(\Lambda,k)\neq false$, then
$F=g_r(\Lambda,k)$ is a solution satisfying (i)--(v). Respectively, for each considered $S\in\mathcal{S}$, the algorithm returns $F=g_r(\Lambda,k)$ of $g_r(\Lambda,k)\neq false$ and returns No. 

This completes the description of our algorithm. To show correctness, recall that the algorithm is trivially correct if $t=1$ or $t>k+1$. Let $2\leq t\leq k+1$. If $|V(G)|\leq 3q$, then we use brute force and, again, the algorithm is trivially correct. Let $|V(G)|>3q$. Then we have that for any solution $F$, $G-F$ has a unique big component  and other  components are small. We consider all $i\in\{1,\ldots,t\}$ and for each $i$, we are looking for a solution such that $G_i$ with $\lambda^{\sf w}(G_i)\geq\lambda_i$ is the big component.
Here we consider two cases. First, we consider the case $\lambda_i>k$, and the correctness of the algorithm is already proved. Then we assume that $\lambda_i\leq k$. Hence, we have to show correctness for this case.

Observe first that whenever the algorithm returns a solution $F$, it is a solution for the instance $(G,w,L,\Lambda,k)$  of \probAWCGC and ${\sf w}(G_i)\geq\lambda_i$.  Therefore, we should prove that if  $(G,w,L,\Lambda,k)$ is a yes-instance, then the algorithm returns a solution $F$. Assume that $(G,w,L,\Lambda,k)$ a yes-instance. Then there is a solution $F$.  
We construct a family $\mathcal{S}$ of at most  $2^{\Oh(q(q+k)\log(q+k))}\cdot \log n$ subsets of $V(G)$ such that $\mathcal{S}$ contains a set $S$ that is $(A,B)$-good with respect to $F$.
Since our algorithm considers all $S\in \mathcal{S}$, for one of the choices, a set $S$ that is $(A,B)$-good with respect to $F$ is considered. Then we apply Reduction Rules~\ref{red:small-s}--\ref{red:stop-two}. Analyzing these rules, we proved that if  $S$ is $(A,B)$-good with respect to $F$, then we do not stop while executing these rules and the modified set $S$ obtained from the original $S$ remains  $(A,B)$-good with respect to $F$. Moreover, we have that $F$ satisfies the conditions (i)--(v). We obtain that the dynamic programming algorithms, whose correctness follows from Claim~D, should find a solution. This completes the correctness proof.

 Now we evaluate the running time. Clearly, it is sufficient to consider the case  $2\leq t\leq k+1$. If $|V(G)|\leq 3q$, then we solve the problem in time $2^{\Oh(k\log q)}$. Let  $|V(G)|> 3q$. Then we check $t\leq k+1$ values of $i\in\{1,\ldots,t\}$, and for each $i$, try find a solution such that $G_i$ is the big component. We consider the cases $\lambda_i>k$ and $\lambda_i\leq k$. For $\lambda_i>k$, the algorithm runs in time $2^{\Oh(k\log q)}\cdot n^{\Oh(1)}$. For the case $\lambda_i\leq k$, we construct in time $2^{\Oh(q(q+k)\log(q+k))}\cdot n\log n$ a family $\mathcal{S}$ of at most  $2^{\Oh(q(q+k)\log(q+k))}\cdot \log n$ subsets of $V(G)$. Then the algorithm analyzes each $S\in\mathcal{S}$. We apply Reduction Rules~\ref{red:small-s}--\ref{red:stop-two}. Reduction Rules~\ref{red:small-s}, \ref{red:sort} and \ref{red:stop-two} can be applied in polynomial time. Reduction Rules~\ref{red:conn} and \ref{red:stop} can be applied in time $2^{\Oh(k\log(q+k))}\cdot n^{\Oh(1)}$. Finally, we run the dynamic programming algorithm in time $2^{\Oh(k\log q)}\cdot n^{\Oh(1)}$. We conclude that the total running time is 
$2^{\Oh(q(q+k)\log(q+k))}\cdot n^{\Oh(1)}$.
\end{proof}

Using Lemma~\ref{lem:annot}, we construct the algorithm for \probBAWCGC for connected $(q,2k)$-unbreakable graphs.

\begin{lemma}\label{lem:bord-unbreak}
\probBAWCGC can be solved in time $2^{q^32^{2^{\Oh(k)}}}\cdot n^{\Oh(1)}$ for connected $(q,2k)$-unbreakable graphs. 
\end{lemma}

\begin{proof}
Let $(G,\mathbf{x},w,L,\Lambda,k)$ be an instance of \probBAWCGC. Recall that $(G,\mathbf{x})$ is an $r$-boundaried connected graph and $r\leq k$. Recall also that $|\Lambda|\leq k+1$.
To solve the problem, we consider every weighted properly $r$-boundaried graph $(H,\mathbf{y})\in \mathcal{H}_{r,2k}$ and every $\hat{\Lambda}=\langle \hat{\lambda}_1,\ldots,\hat{\lambda}_s\rangle\subseteq \Lambda$, and find  the minimum $0\leq \hat{k}\leq k$ such that $((G,\mathbf{x})\oplus_b(H,\mathbf{y}),w,L,\hat{\Lambda},\hat{k})$ is a yes-instance of \probAWCGC and output a solution $F$ for the instance and output $\emptyset$ if $\hat{k}$ does not exist. Recall that each $r$-boundaried graph in $\mathcal{H}_{r,2k}$ has at most $2^{2^{r-1}}+r$  vertices and since $r\leq 4k$, we have that each graph has $2^{2^{\Oh(k)}}$ vertices. Observe that if $G$ is a connected $(q,2k)$-unbreakable graph, then $(G,\mathbf{x})\oplus_b(H,\mathbf{y})$ is $(q+|V(H)|,2k)$-unbreakable. By Lemma~\ref{lem:annot}, we conclude that we can solve \probAWCGC for each instance
$((G,\mathbf{x})\oplus_b(H,\mathbf{y}),w,L,\hat{\Lambda},\hat{k})$ in time $2^{q^32^{2^{\Oh(k)}}}\cdot n^{\Oh(1)}$. 
By Lemma~\ref{lem:size},  $|\mathcal{H}_{r,2k}|=2^{2^{2^{\Oh(r)}}\log k}$ and the family of boundaried graphs $\mathcal{H}_{r,2k}$ can be constructed in time $2^{2^{2^{\Oh(r)}}\log k}$. 
Since $r\leq 2k$, we have that $\mathcal{H}_{r,2k}$ can be constructed in time $2^{2^{2^{\Oh(k)}}}$ and contains  $2^{2^{2^{\Oh(k)}}}$ boundaried graphs. 
For each  $(H,\mathbf{y})\in \mathcal{H}_{r,2k}$, we have at most $2^{k+1}$ $s$-tuples $\hat{\lambda}$ and consider at most $k+1$ values of $\hat{k}$. Therefore, the total running time is 
$2^{q^32^{2^{\Oh(k)}}}\cdot n^{\Oh(1)}$.
\end{proof}

\subsection{Proof of Theorem~\ref{thm:connect}}\label{sec:proof}
We use the results of the previous subsection to prove  Theorem~\ref{thm:connect} for the general case. To do it, we construct an algorithm for \probBAWCGC.

\begin{lemma}\label{lem:connect}
\probBAWCGC can be solved in time $2^{2^{2^{2^{2^{2^{\Oh(k)}}}}}}\cdot n^{\Oh(1)}$.
\end{lemma}

\begin{proof}
We construct a recursive algorithm for \probBAWCGC.
Let $(G,\mathbf{x},w,L,\Lambda,k)$ be an instance of \probBAWCGC. Recall that $(G,\mathbf{x})$ is an $r$-boundaried connected graph and $r\leq k$. Recall also that $\Lambda$ contains at most $k+1$ elements.

We define the constants $p$ and $q$ that are used throughout the proof as follows:
\begin{equation}\label{eq:p}
%p=k2^{^{\binom{2^{2^{4k-1}}+4k}{2}}}(4k)^{^{\binom{2^{2^{4k-1}}+4k}{2}}}+2k
p=2k2^{k+1}(2k+1)^{^{\binom{2^{2^{4k-1}}+4k}{2}}}+4k
\end{equation}
and
\begin{equation}\label{eq:q}
q=2^{2^{p-1}}+p.
\end{equation}
We are going to use $q$ as a part of the unbreakability threshold and the choice of $q$ is going to be explained latter. Now we just observe that $p=2^{2^{2^{\Oh(k)}}}$ and $q=2^{2^{2^{2^{2^{\Oh(k)}}}}}$.

We apply Lemma~\ref{lem:unbreak} and in time $2^{\Oh(k\log(q+k))}\cdot n^3\log n$ either find a $(q,2k)$-good edge separation $(A,B)$ of $G$ or conclude that $G$ is $(2kq,2k)$-unbreakable. 

If $G$ is $(2kq,2k)$-unbreakable, we apply Lemma~\ref{lem:bord-unbreak} and solve the problem in time $2^{q^32^{2^{\Oh(k)}}}\cdot n^{\Oh(1)}$. Note that the running time could be written as $2^{2^{2^{2^{2^{2^{\Oh(k)}}}}}}\cdot n^{\Oh(1)}$.  
Assume from now that we are given a $(q,2k)$-good edge separation $(A,B)$ of $G$. 

Since $|\mathbf{x}|\leq 4k$ and the vertices of $\mathbf{x}$ are separated between $A$ and $B$, either $A$ or $B$ contains at most $2k$ vertices of $\mathbf{x}$. Assume without loss of generality that $|A\cap\mathbf{x}|\leq 2k$. Let $T$ be the set of end-vertices of the edges of $E(A,B)$ in $A$. Clearly, $|T|\leq 2k$. We form a new $\hat{r}$-tuple $\hat{\mathbf{x}}=\langle \hat{x}_1,\ldots,\hat{x}_{\hat{r}}\rangle$ of vertices $A$ from  the vertices of $(A\cap \mathbf{x})\cup T$; note that $\hat{r}\leq 4k$. We consider $\hat{G}=G[A]$ as the $\hat{\mathbf{x}}$-boundaried graph. We set $\hat{L}=L\cap E(\hat{G})$. This way, we obtain the instance $(\hat{G},\hat{\mathbf{x}},w,\hat{L},\Lambda,k)$ of \probBAWCGC.

Now we solve \probBAWCGC for $(\hat{G},\hat{\mathbf{x}},w,\hat{L},\Lambda,k)$. 

If $|V(\hat{G})|\leq 2^q$, we can simply use brute force. For every weighted properly $\hat{r}$-boundaried graph $(H,\mathbf{y})\in \mathcal{H}_{\hat{r},2k}$ and every $\hat{\Lambda}=\langle \hat{\lambda}_1,\ldots,\hat{\lambda}_s\rangle\subseteq \Lambda$, we consider all subsets of $\hat{L}$ of weight at most $k$ and find 
the minimum $0\leq \hat{k}\leq k$ such that $((\hat{G},\hat{\mathbf{x}})\oplus_b(H,\mathbf{y}),w,L,\hat{\Lambda},\hat{k})$ is a yes-instance of \probAWCGC and output a solution $F$ of minimum $\hat{k}$ for the instance or output $\emptyset$ if $\hat{k}$ does not exist. Since $|\hat{L}|\leq |E(\hat{G})|\leq 2^{2q}$ and $\mathcal{H}_{\hat{r},2k}=2^{2^{2^{\Oh(k)}}}$, the total running time is  $2^{2^{2^{2^{2^{2^{\Oh(k)}}}}}}$. 

If $|V(\hat{G})|> 2^q$, we recursively solve \probBAWCGC for $(\hat{G},\hat{\mathbf{x}},w,\hat{L},\Lambda,k)$ by calling our algorithm for the instance that has lesser size, because $|V(\hat{G})|\leq |V(G)|-q$.

By solving \probBAWCGC for $(\hat{G},\hat{\mathbf{x}},\hat{L},\Lambda,k)$, we obtain a list $\mathcal{L}$ of sets where each element is either $\emptyset$ or $F\subseteq \hat{L}$ that is a solution for the corresponding instance of \probAWCGC for some 
$(H,\mathbf{y})\in \mathcal{H}_{\hat{r},2k}$, $\hat{\Lambda}\subseteq\Lambda$ and $\hat{k}\leq k$. Denote by $M$ the union of all the sets in $\mathcal{L}$. Clearly, $M\subseteq \hat{L}$.

We define $L^*=(L\setminus \hat{L})\cup M$. Since $M\subseteq \hat{L}$, $L^*\subseteq L$.
We show that the instances $(G,\mathbf{x},w,L,\Lambda,k)$ and $(G,\mathbf{x},w,L^*,\Lambda,k)$ are essentially equivalent by proving the following claim.

\medskip
\noindent
{\bf Claim~A.} {\it For every weighted properly $r$-boundaried graph $(H,\mathbf{y})\in \mathcal{H}_{r,2k}$, every $\hat{\Lambda}=\langle \hat{\lambda}_1,\ldots,\hat{\lambda}_s\rangle\subseteq \Lambda$ and every nonnegative integer $\hat{k}\leq k$, $((G,\mathbf{x})\oplus_b(H,\mathbf{y}),w,L,\hat{\Lambda},\hat{k})$ is a yes-instance of \probAWCGC if and only if $((G,\mathbf{x})\oplus_b(H,\mathbf{y}),w,L^*,\hat{\Lambda},\hat{k})$ is a yes-instance of \probAWCGC. 
}

\begin{subproof}[Proof of Claim~A] 
Let $(H,\mathbf{y})\in \mathcal{H}_{r,2k}$, $\hat{\Lambda}=\langle \hat{\lambda}_1,\ldots,\hat{\lambda}_s\rangle\subseteq \Lambda$ and $\hat{k}\leq k$. Since $L^*\subseteq L$,  if 
$((G,\mathbf{x})\oplus_b(H,\mathbf{y}),w,L^*,\hat{\Lambda},\hat{k})$ is a yes-instance of \probAWCGC, then $((G,\mathbf{x})\oplus_b(H,\mathbf{y}),w,L,\hat{\Lambda},\hat{k})$ is a yes-instance of \probAWCGC. Hence, the task is  to show the opposite implication.

Suppose that $((G,\mathbf{x})\oplus_b(H,\mathbf{y}),w,L,\hat{\Lambda},\hat{k})$ is a yes-instance of \probAWCGC and let $F\subseteq L$ be a solution. Let $Q=(G,\mathbf{x})\oplus_b(H,\mathbf{y})$. 
%and let $\hat{\Lambda}=\langle \hat{\lambda}_1,\ldots,\hat{\lambda}_s\rangle$. 
We have that $Q-F$ has $s$ components $Q_1,\ldots,Q_s$ with $\lambda^{\sf w}(Q_i)\geq\hat{\lambda}_i$ for $i\in\{1,\ldots,s\}$. 
We say that $Q_i$ is an \emph{$A$-inner} component if $V(Q_i)\subseteq A$ and we say that $Q_i$ is \emph{$A$-boundary} if $V(Q_i)\cap A\neq\emptyset$ and $V(Q_i)\cap (V(Q)\setminus A)\neq\emptyset$. We assume without loss of generality that $Q_1,\ldots,Q_f$ are the $A$-boundary components, $Q_{f+1},\ldots,Q_g$ are the $A$-inner components and $Q_{g+1},\ldots,Q_s$ are the remaining components of $Q-F$, because the ordering is irrelevant for the forthcoming arguments. Note that some of these groups of the components may be empty. Denote $\tilde{\Lambda}=\langle \hat{\lambda}_1,\ldots,\hat{\lambda}_g\rangle$,   $\tilde{F}=F\cap E(\hat{G})$, and let $\tilde{k}={\sf w}(\tilde{F})$. Slightly abusing notation, we do not distinguish the vertices of $\mathbf{x}$ and the corresponding vertices of $\mathbf{y}$ that are identified in $Q$.

For $i\in\{1,\ldots,f\}$, denote by $X_i$ the set of vertices of $Q_i$ that are in $A$ and have neighbors outside $A$. Let $Q_i'=Q_i[A\cap V(Q_i)]$ and $Q_i''=Q_i[X_i\cup (V(Q_i)\setminus A)]-E(Q_i[X_i])$. Observe that $X_1,\ldots,X_f$ are disjoint subsets of $\hat{\mathbf{x}}$. 
Respectively, each $X_i$ can be seen as an $|X_i|$-subtuple of $\hat{\mathbf{x}}$ and we have that $Q_i=(Q_i',X_i)\oplus_b(Q_i'',X_i)$ for $i\in\{1,\ldots,f\}$. Notice that $Q_1'',\ldots,Q_f''$ are proper boundaried graphs. 

We claim that for each $i\in\{1,\ldots,f\}$, for every component $C$ of $Q_i''$, there is $v\in V(C)\setminus X_i$ such that $\lambda^{\sf w}(X_i,v)\leq 2k$. To show this, we consider two cases. If $(V(C)\setminus X_i)\cap (V(H)\setminus\mathbf{y})\neq\emptyset$, then because $L\cap E(H)=\emptyset$, there is a component $C'$ of $H$ such that $V(C')\subseteq V(C)$ and $V(C')\setminus \mathbf{y}\neq\emptyset$. Then by the definition of $\mathcal{H}_{r,2k}$, there is $v\in V(C')\setminus\mathbf{y}$ such that $\lambda^{\sf w}(\mathbf{y},v)\leq 2k$. Therefore, we obtain that $2k\geq \lambda^{\sf w}(\mathbf{y},v)\geq \lambda^{\sf w}(X_i,v)$. Assume that $C\cap (V(H)\setminus\mathbf{y})=\emptyset$. Then $V(C)\setminus X_i\subseteq B$ and for every $v\in V(C)\setminus X_i$, $\lambda^{\sf w}(X_i,v)\leq {\sf w}(A,B)\leq 2k$. Note that, in particular, we have that $\lambda^{\sf w}(Q_i)\leq 2k$ for $i\in\{1,\ldots,f\}$. 

Denote by $R$ the weighted graph with the vertex set $\hat{\mathbf{x}}\cup(\bigcup_{i=1}^fV(Q_i''))$ and the edge set $\bigcup_{i=1}^fE(Q_i'')$ where the weights of the edges are inherited from the weights in $G$. Observe that $(R,\hat{\mathbf{x}})$ is a weighted properly $\hat{r}$-boundaried graph. Let $(R',\hat{\mathbf{x}})$ be the boundaried  graph obtained from $(R,\hat{\mathbf{x}})$ by the cut reduction with respect to $2k$.

By Lemma~\ref{lem:cut-additional} (i), $R$ has at most $2^{2^{\hat{r}-1}}+\hat{r}$ vertices. Since $Q_1'',\ldots,Q_f''$ are the components of $R$ that have at least one vertex outside $\hat{\mathbf{x}}$, we obtain that by  Lemma~\ref{lem:cut-additional} (ii), for each component $C$ of $R'$ that has at least one vertex outside $\hat{\mathbf{x}}$, there is
$v\in V(C)\setminus\mathbf{y}$ such that $\lambda^{\sf w}(R',\hat{\mathbf{x}},v)\leq 2k$. This implies that $\mathcal{H}_{\hat{r},2k}$ contains a weighted properly $\hat{r}$-boundaried graph that is isomorphic to $(R',\hat{\mathbf{x}})$. To simplify notations, we assume that $(R',\hat{\mathbf{x}})\in \mathcal{H}_{\hat{r},2k}$.

Consider the graph $(\hat{G}-\tilde{F},\hat{x})\oplus_b(R,\hat{x})$. This graph has the components $Q_1,\ldots,Q_g$ and $\lambda^{\sf w}(Q_i)\geq\lambda_i$ for $i\in\{1,\ldots,g\}$. By Lemma~\ref{lem:replacement}, we have that $(\hat{G}-\tilde{F},\hat{x})\oplus_b(R',\hat{x})$ has $g$ components $Q_1',\ldots,Q_g'$ such that it holds $\lambda^{\sf w}(Q_i')\geq\lambda_i$ for $i\in\{1,\ldots,g\}$. This immediately implies that  
$(G',w,\hat{L},\tilde{\Lambda},\tilde{k})$ is a yes-instance of \probAWCGC for $G'=(\hat{G},\hat{\mathbf{x}})\oplus_b(R',\hat{\mathbf{x}})$ with $\tilde{F}$ being a solution. 
 Notice that the algorithm for  \probBAWCGC for $(\hat{G},\hat{\mathbf{x}},\hat{L},\Lambda,k)$ solves the instance 
$(G',w,\hat{L},\tilde{\Lambda},\tilde{k})$  of \probAWCGC and, since we have a yes-instance, the output $\mathcal{L}$ of the algorithm contains a solution $\hat{F}\subseteq M$. 
Therefore,  $(\hat{G}-\hat{F},\hat{x})\oplus_b(R',\hat{x})$ has $g$ components $S_1',\ldots,S_g'$ such that it holds $\lambda^{\sf w}(S_i')\geq\lambda_i$ for $i\in\{1,\ldots,g\}$.
By Lemma~\ref{lem:replacement} we obtain that $(\hat{G}-\hat{F},\hat{x})\oplus_b(R,\hat{x})$ has $g$ components  $S_1,\ldots,S_g$ such that it holds $\lambda^{\sf w}(S_i)\geq\lambda_i$ for $i\in\{1,\ldots,g\}$. Consider $F^*=(F\setminus \tilde{F})\cup \hat{F}$. Recall that $\tilde{F}=F\cap E(\hat{G})$. This immediately implies that $F^*\subseteq L^*$. Also we have that ${\sf w}(\hat{F})\leq\tilde{k}={\sf w}(\tilde{F})$. Hence, ${\sf w}(F^*)\leq {\sf w}(F)\leq \hat{k}$.
We have that $Q-F^*$ has the components $S_1,\ldots,S_g$ and $Q_{g+1},\ldots,Q_s$,  $\lambda^{\sf w}(S_i)\geq\lambda_i$ for $i\in\{1,\ldots,g\}$ and $\lambda^{\sf w}(Q_i)\geq\lambda_i$ for $i\in\{g+1,\ldots,s\}$. We conclude that $F^*$ is a solution for the instance $((G,\mathbf{x})\oplus_b(H,\mathbf{y}),w,L^*,\hat{\Lambda},\hat{k})$ of \probAWCGC, that is,
$((G,\mathbf{x})\oplus_b(H,\mathbf{y}),w,L^*,\hat{\Lambda},\hat{k})$ is a yes-instance of \probAWCGC.
\end{subproof}

By Claim~A we obtain that  every solution of $(G,\mathbf{x},w,L^*,\Lambda,k)$ is a solution of $(G,\mathbf{x},w,L,\Lambda,k)$, and there is a solution of 
$(G,\mathbf{x},w,L,\Lambda,k)$ that is a solution of  $(G,\mathbf{x},w,L^*,\Lambda,k)$. Therefore, it is sufficient for us to solve $(G,\mathbf{x},w,L^*,\Lambda,k)$.

Let $Z\subseteq A$ be the set of end-vertices of the edges of $M$ and the vertices of $\hat{\mathbf{x}}$. Because $\hat{r}\leq 4k$, $\mathcal{H}_{\hat{r},2k}\leq  (2k+1)^{^{\binom{2^{2^{4k-1}}+4k}{2}}}$ by Lemma~\ref{lem:size}. Since $t\leq k+1$, there are at most $2^{k+1}$ subtuples $\hat{\Lambda}$ of $\Lambda$. For each $(H,y)\in\mathcal{H}_{\hat{r},2k}$ and $\hat{\Lambda}\subseteq\Lambda$, the solution $\mathcal{L}$ of \probBAWCGC for $(\hat{G},\hat{\mathbf{x}},\hat{L},\Lambda,k)$ contains a set $F$ of size at most $k$. This implies that 
$$|M|\leq k2^{k+1}(2k+1)^{^{\binom{2^{2^{4k-1}}+4k}{2}}}.$$
Because $|\hat{\mathbf{x}}|\leq 4k$, we obtain that 
\begin{equation}\label{eq:Z}
|Z|\leq 2|M|+4k\leq 2k2^{k+1}(2k+1)^{^{\binom{2^{2^{4k-1}}+4k}{2}}}+4k=p
\end{equation}
for $p$  defined in (\ref{eq:p}).

Let $U=A\setminus Z$. We define $Q=G- U$. Let also $R$ be the graph with the vertex set $A$ and the edge set $E(G[A])\setminus E(G[Z])$. We order the vertices of $Z$ arbitrarily and consider $Z$ to be $|Z|$-tuple of the vertices of $Q$ and $R$. Observe that $(R,Z)$ is a properly $|Z|$-boundaried graph as $G[A]$ is connected. Since $V(F)\cap V(R)=Z$, we have that 
$G=(Q,Z)\oplus_b(R,Z)$. Let $(R^*,Z)$ be the boundaried graph obtained from  $(R,Z)$ by the cut reduction with respect to $+\infty$. We define $G^*=(Q,Z)\oplus_b(R^*,Z)$. 
Note that $L^*\subseteq E(Q)\subseteq E(G^*)$. We show that we can replace $G$ by $G^*$ in the considered instance  $(G,\mathbf{x},w,L^*,\Lambda,k)$ of \probBAWCGC.

\medskip
\noindent
{\bf Claim~B.} {\it For every weighted properly $r$-boundaried graph $(H,\mathbf{y})\in \mathcal{H}_{r,2k}$, every $\hat{\Lambda}=\langle \hat{\lambda}_1,\ldots,\hat{\lambda}_s\rangle\subseteq \Lambda$ and every nonnegative integer $\hat{k}\leq k$, a set $F\subseteq L^*$ is a solution for the instance $((G,\mathbf{x})\oplus_b(H,\mathbf{y}),w,L^*,\hat{\Lambda},\hat{k})$  if and only if $F$ is a solution for $((G^*,\mathbf{x})\oplus_b(H,\mathbf{y}),w,L^*,\hat{\Lambda},\hat{k})$. 
}

\begin{subproof}[Proof of Claim~B]
Notice that $\mathbf{x}\subseteq V(Q)$ by the construction of $F$ and $R$, Hence, 
$$(G,\mathbf{x})\oplus_b(H,\mathbf{y})=((Q,\mathbf{x})\oplus_b(H,\mathbf{y}),Z)\oplus_b(R,Z)$$
and
$$(G^*,\mathbf{x})\oplus_b(H,\mathbf{y})=((Q,\mathbf{x})\oplus_b(H,\mathbf{y}),Z)\oplus_b(R^*,Z).$$
For every $F\subseteq L^*\subseteq E(Q)$, we have that 
$$\tilde{G}=(G,\mathbf{x})\oplus_b(H,\mathbf{y})-F=((Q-F,\mathbf{x})\oplus_b(H,\mathbf{y}),Z)\oplus_b(R,Z)$$
and
$$\tilde{G}^*=(G^*,\mathbf{x})\oplus_b(H,\mathbf{y})-F=((Q-F,\mathbf{x})\oplus_b(H,\mathbf{y}),Z)\oplus_b(R^*,Z).$$
By Lemma~\ref{lem:replacement}, we have that $\tilde{G}$ has components $G_1,\ldots,G_s$  such that $\lambda^{\sf w}(G_i)\geq\hat{\lambda}_i$ for $i\in\{1,\ldots,s\}$ if and only if the same holds for $\tilde{G}^*$, that is,  $\tilde{G}^*$ has components $G_1',\ldots,G_s'$  such that $\lambda^{\sf w}(G_i')\geq\hat{\lambda}_i$ for $i\in\{1,\ldots,s\}$, and this proves the claim.
\end{subproof} 

By Claim~B, to solve \probBAWCGC for $(G,\mathbf{x},w,L^*,\Lambda,k)$, it is sufficient to solve the problem for $(G^*,\mathbf{x},w,L^*,\Lambda,k)$.
Observe that $|V(G^*)|=|B|+|V(R^*)|$. Because  $(R^*,Z)$ is obtained by the cut reduction, we have that 
$|V(R^*)|\leq 2^{2^{|Z|-1}}+|Z|$ by Lemma~\ref{lem:cut-additional} (i). Using (\ref{eq:Z}), we have that 
$$|V(R^*)|\leq  2^{2^{p-1}}+p=q$$
for $q$ defined in (\ref{eq:q}). Recall that $|A|>q$ since $(A,B)$ is a $(q,2k)$-good edge separation of $G$.
Therefore,
$$|V(G^*)|=|B|+|V(R^*)|\leq |B|+q<|A|+|B|=|V(G)|.$$
We use this and solve \probBAWCGC for $(G^*,\mathbf{x},w,L^*,\Lambda,k)$ recursively by applying our recursive algorithm for the instance with the input graph of smaller size.

This completes the description of the algorithm for  \probBAWCGC and its correctness proof. Now we evaluate the running time. 
Denote by $\tau(G,\mathbf{x},w,L,\Lambda,k)$ the time needed to solve \probBAWCGC for $(G,\mathbf{x},w,L,\Lambda,k)$. 
Recall that we first apply Lemma~\ref{lem:unbreak} and in time $2^{\Oh(k\log(q+k))}\cdot n^3\log n$ either find a $(q,2k)$-good edge separation $(A,B)$ of $G$ or conclude that $G$ is $(2kq,2k)$-unbreakable. Since $q=2^{2^{2^{2^{2^{\Oh(k)}}}}}$, this is done in time $2^{2^{2^{2^{2^{\Oh(k)}}}}}$.
If $G$ is $(2kq,2k)$-unbreakable, we apply Lemma~\ref{lem:bord-unbreak} and solve the problem in time  $2^{2^{2^{2^{2^{2^{\Oh(k)}}}}}}\cdot n^{\Oh(1)}$, that is, in this case
\begin{equation}\label{eq:unbr}
\tau(G,\mathbf{x},w,L,\Lambda,k)=2^{2^{2^{2^{2^{2^{\Oh(k)}}}}}}\cdot n^{\Oh(1)}.
\end{equation}
Otherwise, if we obtain a $(q,2k)$-good edge separation $(A,B)$ of $G$, we solve \probBAWCGC for $(\hat{G},\hat{\mathbf{x}},\hat{L},\Lambda,k)$ in time 
$\tau(\hat{G},\hat{\mathbf{x}},\hat{L},\Lambda,k)$. We obtain a solution $\mathcal{L}$. Recall that $|\mathcal{L}|$ contains sets constructed for each boundaried graph
in $\mathcal{H}_{\hat{r},2k}$ and each $\hat{\Lambda}\subseteq \Lambda$. Combining Lemma~\ref{lem:size} and the fact that $|\lambda|\leq k+1$, we obtain 
that  in  time $2^{2^{2^{\Oh(k)}}}\cdot n^{\Oh(1)}$ we can construct $M$ and $L^*$. 
Given these sets, we construct  $(R,Z)$ in polynomial time. Then 
by Lemma~\ref{lem:cut-additional}, $(R^*,Z)$ is constructed in time  $2^{\Oh(|Z|)}\cdot n^{\Oh(1)}$. Using  (\ref{eq:p}) and (\ref{eq:Z}), we conclude that it can be done in time $2^{2^{2^{\Oh(k)}}}\cdot n^{\Oh(1)}$. Then in a polynomial time we construct the instance $(G^*,\mathbf{x},w,L^*,\Lambda,k)$ and  \probBAWCGC is solved for this instance in time $\tau(G^*,\mathbf{x},w,L^*,\Lambda,k)$. This gives us the recurrence 
\begin{equation}\label{eq:rec}
\tau(G,\mathbf{x},w,L,\Lambda,k)=\tau(\hat{G},\hat{\mathbf{x}},\hat{L},\Lambda,k)+\tau(G^*,\mathbf{x},w,L^*,\Lambda,k)+2^{2^{2^{2^{2^{\Oh(k)}}}}}.
\end{equation}
Recall that $|V(G^*)|=|B|+|V(R^*)|$ and $|V(R^*)|\leq q$. Hence,
\begin{equation}\label{eq:size-recursion}
|V(G^*)|\leq |V(G)|-|V(\hat{G})|+q.
\end{equation}
Combining (\ref{eq:unbr})--(\ref{eq:size-recursion}), we obtain that the  running time of our algorithm is  $2^{2^{2^{2^{2^{2^{\Oh(k)}}}}}}\cdot n^{\Oh(1)}$ applying the general scheme from the paper of Chitnis et al.~\cite{ChitnisCHPP16}.
\end{proof}

Now we are ready to complete the proof Theorem~\ref{thm:connect}. For convenience, we restate the theorem. 

\medskip
\noindent
{\bf Theorem~\ref{thm:connect}.} {\it
\probWCGC can be solved in time   $2^{2^{2^{2^{2^{2^{\Oh(k)}}}}}}\cdot n^{\Oh(1)}$
 if the input graph is connected.
}

\begin{proof}
Let $(G,w,\Lambda,k)$ be an instance \probWCGC where $G$ is connected. If $t=|\Lambda|>k+1$, we conclude that  $(G,w,\Lambda,k)$ is a no-instance, because the connected graph $G$ cannot be partitioned into more than $k+1$ components by deleting at most $k$ edges. In this case we return {\sf NO} and stop. If $t=1$, then we verify in polynomial whether $G$ is weight $\lambda$-connected for the unique element $\lambda$ of $\Lambda$ using, e.g, the algorithm of Stoer and Wagner~\cite{StoerW97}. Assume that these are not the cases, that is, $2\leq t\leq k+1$. 

We construct the instance $(G,\mathbf{x},w,L,\Lambda,k)$ of \probBAWCGC by setting $L=E(G)$ and defining $\mathbf{x}=\emptyset$, i.e., we consider $(G,\mathbf{x})$ be a $0$-boundaried graph. then we solve \probBAWCGC for $(G,\mathbf{x},w,L,\Lambda,k)$ in time $2^{2^{2^{2^{2^{2^{\Oh(k)}}}}}}\cdot n^{\Oh(1)}$ using Lemma~\ref{lem:connect}. It remains to observe that 
$(G,w,\Lambda,k)$ is a yes instance \probWCGC if and only if the output produced by the algorithm for \probAWCGC contains  nonempty set for the unique empty $0$-boundaried graph in $\mathcal{H}_{0,2k}$ and $\hat{\Lambda}=\Lambda$.
\end{proof}

\section{The algorithm for \probWCGClong}\label{sec:general}
In this section we extend the result obtained in Section~\ref{sec:connect} and show that \probWCGC is \classFPT when parameterized by $k$ even if the input graph is disconnected. 
%First, we introduce some additional notations.

%Let $\alpha=\langle \alpha_1,\ldots,\alpha_t\rangle$ where $\alpha_i\in\mathbb{N}\cup\{+\infty\}$ for $i\in\{1,\ldots,t\}$ and  $\alpha_1\leq\ldots\leq \alpha_t$.
%We call the \emph{variate} of $\alpha$ the set of distinct elements of $\alpha$ and denote it by $\mathbf{var}(\alpha)$. Let also  $\beta=\langle \beta_1,\ldots,\beta_s\rangle$ where $\beta_i\in\mathbb{N}\cup\{+\infty\}$ for $i\in\{1,\ldots,s\}$ and  $\beta_1\leq\ldots\leq \beta_t$. We say that $\gamma=\alpha+\beta$ is the \emph{merge} of $\alpha$ and $\beta$ if $\gamma$ is the $(t+s)$-tuple obtained by sorting the elements of $\alpha$ and $\beta$ in the increasing order. We denote by $r\alpha$ the merge of $r$ copies of $\alpha$.
%If $s=t$, we write $\alpha\preceq\beta$ if $\alpha_i\leq\beta_i$ for $i\in\{1,\ldots,t\}$.
% 

% and we write $\mathcal{P}_{\leq p}(n)$ to denote the set of all $q$-partitions of $n$ for all positive $q\leq p$.

First, we solve \probWCGC for the case when all the components of the input graph have the same weighted connectivity.

\begin{lemma}\label{lem:uniform}
\probWCGC can be solved in time   $2^{2^{2^{2^{2^{2^{\Oh(k)}}}}}}\cdot n^{\Oh(1)}$ if for every component $C$ of the input graph, $\lambda^{\sf w}(C)=\lambda\leq k$.
\end{lemma}

\begin{proof}
Let $(G,w,\Lambda,k)$ be an instance \probWCGC. Let $G_1,\ldots,G_s$ be the components of $G$ and $\lambda^{\sf w}(G_i)=\lambda\leq k$ for $i\in\{1,\ldots,s\}$. Let also
 $\Lambda=\langle \lambda_1,\ldots,\lambda_t\rangle$. 

%If $G$ is connected, then we solve the problem using Theorem~\ref{thm:connect}. Assume that this is not the case.  
If $s>t$, then $(G,w,\Lambda,k)$ is a trivial no-instance of  \probWCGC since we cannot reduce the number of components by deleting edges. If $s=t$, then $(G,w,\Lambda,k)$ is a yes-instance if and only if $\Lambda\leq t\langle \lambda\rangle$
and we verify this condition in polynomial time. If $s<t-k$,  then $(G,w,\Lambda,k)$ is no-instance, because by deleting at most $k$ edges it is possible to obtain at most $k$ additional components. In all these cases we return the corresponding answer and stop. From now we assume that $t-k\leq s\leq t$.

Observe that if $|\mathbf{var}(\Lambda)|>3k$, then $(G,w,\Lambda,k)$ is a no-instance of \probWCGC. Indeed, if $|\mathbf{var}(\Lambda)|>3k$, then $\Lambda$ contains at least $2k+1$ elements $\lambda_i>k$. Since $\lambda^{\sf w}(G_i)=\lambda\leq k$ for $i\in\{1,\ldots,s\}$, if $F\subseteq E(G)$ is a solution, then $G-F$ should have at least $2k+1$ components with weighted connectivity at least $k+1$, but the deletion of at most $k$ edges can create at most $2k$ graphs with weighted connectivity at least $k+1$ from the components with weighted connectivity at most $k$.   Hence, if $|\mathbf{var}(\Lambda)|>3k$, we return {\sf NO} and stop. Assume that $|\mathbf{var}(\Lambda)|\leq 3k$.

For $F\subseteq E(G)$, denote by $I_F=\{i\in\{1,\ldots,s\}\mid E(G_i)\cap F\neq\emptyset\}$.

\medskip
\noindent 
{\bf Claim~A.} {\it
The instance $(G,w,\Lambda,k)$ is a yes-instance of \probWCGC if and only if there is $F\subseteq E(G)$ with ${\sf w}(F)\leq k$ such that $I_F=\{j_1,\ldots,j_p\}$ where $p\leq k$ and the following is fulfilled: 
for each $i\in \{1,\ldots,p\}$, there is a $t_i$-subtuple $\Lambda^i=\langle \lambda_1^i,\ldots,\lambda_{t_i}^i\rangle\subseteq \Lambda$ such that 
\begin{itemize}
\item[(i)] $t_1+\ldots+t_p=t-s+p$,
\item[(ii)] $G_{j_i}-(F\cap E(G_{j_i}))$ has $t_i$ components $G_{j_i}^1,\ldots,G_{j_i}^{t_i}$ with $\lambda^{\sf w}(G_{j_i}^h)\geq \lambda_j^h$ for $h\in\{1,\ldots,t_i\}$,
\end{itemize}
and $\Lambda\preceq(s-p)\langle \lambda\rangle+\Lambda^1+\ldots+\Lambda^p$.}

\begin{subproof}[Proof of Claim A]
To see that the claim holds, it is sufficient to observe that any solution $F$ for  $(G,w,\Lambda,k)$ should split $p\leq k$ components of $G$ into $t-s+p$ components in such a way that for the obtained components together with the old nonsplit components, the connectivity constraints are fulfilled. Vice versa, if we manage to split  $p\leq k$ components of $G$ into $t-s+p$ components by deleting a set of edges $F$ with ${\sf w}(F)\leq k$ in such a way that for the obtained components together with the old nonsplit components the connectivity constraints are fulfilled, then  
$(G,w,\Lambda,k)$ is a yes-instance.
\end{subproof}

\medskip
We use Claim A to solve the problem. 

For a positive integer $r$, denote by $\mathcal{L}_r$ the set of $r$-subtuples $\hat{\Lambda}\subseteq\Lambda$.  
Recall that for positive integers $q$ and $p$, a \emph{$p$-partition} of $q$ as a $p$-tuple $\langle r_1,\ldots,r_p\rangle$ of positive integers such that $r_1\leq\ldots\leq r_p$ and $q=r_1+\ldots+r_p$. We demote by $\mathcal{P}_p(q)$ the set of all $p$-partitions of $q$.

We consider all positive integers $p\leq k$ and all partitions $\langle t_1,\ldots,t_p\rangle\in \mathcal{P}_{p}(t-s+p)$.
For every $\langle t_1,\ldots,t_p\rangle$, we consider all possible $p$-tuples $\langle\Lambda^{1},\ldots,\Lambda^{p}\rangle$ for $\Lambda^i\in \mathcal{L}_{t_i}$ for $i\in\{1,\ldots,p\}$ such that  $\Lambda\preceq(s-p)\langle \lambda\rangle+\Lambda^1+\ldots+\Lambda^p$.

For each choice of $(\langle t_1,\ldots,t_p\rangle,\langle\Lambda^{1},\ldots,\Lambda^{p}\rangle)$, we construct the auxiliary weighted complete bipartite graph $\mathcal{G}$ with the vertex set $\{1,\ldots,p\}\cup\{G_1,\ldots,G_s\}$ where $\{1,\ldots,p\}$ and $\{G_1,\ldots,G_s\}$ form the bipartition. For every $i\in\{1,\ldots,p\}$ and $j\in\{1,\ldots,s\}$, we define the \emph{cost} $c(iG_j)$ of 
$iG_j$ as follows. We use Theorem~\ref{thm:connect} to find the minimum $\hat{k}\leq k$ such that $(G_j,w,\Lambda^i,\hat{k})$ is a yes-instance of  \probWCGC. We set $c(iG_j)=\hat{k}$ if we find such a value and $c(iG_j)=+\infty$ if $\hat{k}\leq k$ does not exist.
We find a matching $M$ of minimum cost that saturates the vertices $\{1,\ldots,p\}$ using, e.g, the Hungarian method~\cite{Kuhn55,Kuhn56}.
If for the cost of $M$ it holds that $c(M)\leq k$, we return the answer {\sf YES} for $(G,w,\Lambda,k)$
and stop. Otherwise, we discard the current choice of  $(\langle t_1,\ldots,t_p\rangle,\langle\Lambda^{1},\ldots,\Lambda^{p}\rangle)$.

If we do not return {\sf YES} for any choice of  $(\langle t_1,\ldots,t_p\rangle,\langle\Lambda^{1},\ldots,\Lambda^{p}\rangle)$, we return {\sf NO} and stop.

\medskip
To show correctness, assume first that $(G,w,\Lambda,k)$ is a yes-instance of \probWCGC. Then by Claim~A, there is $F\subseteq E(G)$ with ${\sf w}(F)\leq k$ such that $I_F=\{j_1,\ldots,j_p\}$ where $p\leq k$ and the following is fulfilled: 
for each $i\in \{1,\ldots,p\}$, there is a $t_i$-subtuple $\Lambda^i=\langle \lambda_1^i,\ldots,\lambda_{t_i}^i\rangle\subseteq \Lambda$ such that (i) and (ii) are fulfilled and 
and $\Lambda\preceq(s-p)\langle \lambda\rangle+\Lambda^1+\ldots+\Lambda^p$. We can assume without loss of generality that %$I_F=\{1,\ldots,p\}$ and 
$t_1\leq\ldots\leq t_p$ as, otherwise, we can reorder the components of $G$.  Clearly, $\langle t_1,\ldots,t_p\rangle\in \mathcal{P}_{p}(t-s)(n)$ and $\Lambda^i\in \mathcal{L}_{t_i}$ for $i\in\{1,\ldots,p\}$. Then by (ii), $(G_{j_i},w,\Lambda^i,\hat{k})$ is a yes-instance of  \probWCGC for $\hat{k}={\sf w}(F\cap E(G_{j_i}))$, that is, $\sum_{i=1}^pc(iG_{j_i})\leq k$. We conclude that 
$M=\{iG_{j_i}\mid i\in\{1,\ldots,p\}\}$ is a matching of cost at most $k$ in $\mathcal{G}$.

Assume now that for some choice of $(\langle t_1,\ldots,t_p\rangle,\langle\Lambda^{1},\ldots,\Lambda^{p}\rangle)$, $\mathcal{G}$ has a matching $M=\{iG_{j_i}\mid i\in\{1,\ldots,p\}\}$  of cost at most $k$. For every $i\in\{1,\ldots,p\}$, $(G_{j_i},w,\Lambda^i,\hat{k})$ is a yes-instance of  \probWCGC for $\hat{k}=c(iG_{j_i})$. Let $F_i$ be a solution for this instance and let $F=\bigcup_{i=1}^pF_i$. Then Claim~A immediately implies that $(G,w,\Lambda,k)$ is a yes-instance of \probWCGC.

Since  the preliminary steps of the algorithm can be done in polynomial time, we have to evaluate the running time of this final part of the algorithm. Observe that because $p\leq k$ and $t-s\leq k$, we consider $2^{\Oh(k\log k)}$ partitions $\langle t_1,\ldots,t_p\rangle\in \mathcal{P}_{p}(t-s+p)$.
Since each $t_i\leq 2k$ and $\mathbf{var}(\Lambda)\leq 3k$, there are $2^{\Oh(k\log k)}$ choices of each $\Lambda^i$ in  $\langle\Lambda^{1},\ldots,\Lambda^{p}\rangle$, and 
$2^{\Oh(k^2\log k)}$ choices of $\langle\Lambda^{1},\ldots,\Lambda^{p}\rangle$. That is, we consider $2^{\Oh(k^2\log k)}$ choices of $(\langle t_1,\ldots,t_p\rangle,\langle\Lambda^{1},\ldots,\Lambda^{p}\rangle)$. Note that we can verify whether  $\Lambda\preceq(s-p)\langle \lambda\rangle+\Lambda^1+\ldots+\Lambda^p$ in polynomial time. The construction of $\mathcal{G}$ together with computing the costs of edges can be done in time $2^{2^{2^{2^{2^{2^{\Oh(k)}}}}}}\cdot n^{\Oh(1)}$ as we solve at most $kps$ instances of \probWCGC for connected graphs for $p\leq k$ and $s\leq n$ using Theorem~\ref{thm:connect}. Then the total running time is  $2^{2^{2^{2^{2^{2^{\Oh(k)}}}}}}\cdot n^{\Oh(1)}$.
\end{proof}

Now we are ready to show the main claim.

\begin{theorem}\label{thm:general}
\probWCGC can be solved in time   $2^{2^{2^{2^{2^{2^{\Oh(k)}}}}}}\cdot n^{\Oh(1)}$.
\end{theorem}

\begin{proof}
Let $(G,w,\Lambda,k)$ be an instance \probWCGC, $\Lambda=\langle \lambda_1,\ldots,\lambda_t\rangle$. 

%If $G$ is connected, then we solve the problem using Theorem~\ref{thm:connect}. Assume that this is not the case.
We find the components of $G$ and compute their weighted connectivities using the algorithm of Stoer and Wagner~\cite{StoerW97} for finding minimum cuts. Assume that $G_1,\ldots,G_s$ are the components of $G$ and $\lambda^{\sf w}(G_1)\leq\ldots\leq\lambda^{\sf w}(G_s)$. If $s>t$, then $(G,w,\Lambda,k)$ is a trivial no-instance of  \probWCGC since we cannot reduce the number of components by deleting edges. If $s=t$, then $(G,w,\Lambda,k)$ is a yes-instance if and only if $\Lambda\preceq\langle\lambda^{\sf w}(G_1),\ldots,\lambda^{\sf w}(G_t)\rangle$  and we verify this condition in polynomial time. If $s<t-k$,  then $(G,w,\Lambda,k)$ is no-instance, because by deleting at most $k$ edges it is possible to obtain at most $k$ additional components. In all these cases we return the corresponding answer and stop. From now we assume that $t-k\leq s< t$.

We exhaustively apply the following reduction rule.

\begin{reduction}\label{red:remove}
If there is $i\in \{1,\ldots,s\}$ such that $\lambda^{\sf w}(G_i)>k$, then find the maximum $j\in\{1,\ldots,t\}$ such that $\lambda^{\sf w}(G_i)\geq\lambda_j$ and 
set $G=G-V(G_i)$ and $\Lambda=\Lambda\setminus\{\lambda_j\}$.
\end{reduction}

To see that the reduction rule is safe, that is, it produces an equivalent instance of \probWCGC, assume that  $\lambda^{\sf w}(G_i)>k$ for a component of $G$. Let $G'=G-V(G_i)$ and $\Lambda'=\Lambda\setminus\{\lambda_j\}$ be obtained by the application of the rule. 
For any $S\subseteq E(G_i)$ with ${\sf w}(S)\leq k$, $G_i-S$ is connected. This implies that if $(G,w,\Lambda,k)$ is a yes-instance, then it has a solution $F\subseteq E(G)$ with $F\cap E(G_i)=\emptyset$, that is, $G_i$ is a component of $G-F$.  Observe that by the choice of $\lambda_j$, $\lambda_j$ is the maximum connectivity constraint satisfied by $G_i$. Then 
$(G',w,\Lambda',k)$ is a yes-instance of \probWCGC. For the opposite direction, assume that $F$ is a solution for $(G',w,\Lambda',k)$. Then it is straightforward to verify that $F$ is a solution for $(G,w,\Lambda,k)$ as well.

To simplify notations, assume that $G$ with its components $G_1,\ldots,G_s$ and $\Lambda=\langle \lambda_1,\ldots,\lambda_t\rangle$ is obtained from the original input by the 
 exhaustive application of Reduction Rule~\ref{red:remove}. Note that we still have that  $t-k\leq s< t$. It may happen that $s=0$, that is, $G$ became empty after the application of the rule. In this case $(G,w,\Lambda,k)$ is a trivial no-instance, and we return {\sf NO} and stop. Assume that $s\geq 1$.
Observe that we obtain that $\lambda^{\sf w}(G_i)\leq k$ for $i\in\{1,\ldots,s\}$, because Reduction Rule~\ref{red:remove} cannot be applied any more.  

In exactly the same way as in the proof of Lemma~\ref{lem:uniform} we have that  if $|\mathbf{var}(\Lambda)|>3k$, then $(G,w,\Lambda,k)$ is a no-instance of \probWCGC. If  $|\mathbf{var}(\Lambda)|>3k$, then $\Lambda$ contains at least $2k+1$ elements $\lambda_i>k$ and it is impossible to obtain at least $2k+1$ components from the components with weighted connectivity at most $k$ by deleting a set of at most $k$ edges.  Hence, if $|\mathbf{var}(\Lambda)|>3k$, we return {\sf NO} and stop. Assume that $|\mathbf{var}(\Lambda)|\leq 3k$.

For $i\in\{1,\ldots,k\}$, denote by $H_i$ the graph obtained by taking the union of the components $G_j$ of $G$ of weighted connectivity $\lambda^{\sf w}(G_j)=i$. Note that some of $H_i$ could be empty. Let $s_i$  denote the number of components of $H_i$ for $i\in\{1,\ldots,k\}$. 
Similarly to Claim~A in the proof of Lemma~\ref{lem:uniform} we can observe the following.

\medskip
\noindent 
{\bf Claim~B.} {\it
The instance $(G,w,\Lambda,k)$ is a yes-instance of \probWCGC if and only if there are nonnegative integers $h_1,\ldots,h_k$, $p_1,\ldots,p_k$ and $t_1,\ldots,t_k$, and $t_i$-subtuples $\Lambda^i\subseteq\Lambda$ 
such that 
\begin{itemize}
\item[(i)] $h_1+\ldots+h_k\leq k$,
\item[(ii)] $p_i\leq s_i$, $p_i\leq t_i\leq 2h_i$ for $i\in\{1,\ldots,k\}$,
\item[(iii)]  $\sum_{i=1}^k(t_i-p_i)=t-s$,
\item[(iv)] for each $i\in\{1,\ldots,k\}$, $(H_i,w,(s_i-p_i)\langle i\rangle+\Lambda^i,h_i)$ is a yes-instance of \probWCGC, and 
\end{itemize}
and $\Lambda\preceq\sum_{i=1}^k((s_i-p_i)\langle \lambda\rangle+\Lambda^1+\ldots+\Lambda^p)$.
}

\begin{subproof}[Proof of Claim~B]
Suppose that $F$ is a solution of $(G,w,\Lambda,k)$. We assume that $F$ is inclusion minimal, that is, each edge of $F$ is included in some separator. For $i\in\{1,\ldots,k\}$, let $h_i={\sf w}(F\cap E(H_i))$ and $p_i$ be the number of components of $H_i$ containing edges of $F$. Then the deletion of the edges of $F\cap E(H_i)$ separates $p_i$ components of $H_i$ into $t_i$ components $C_1^i,\ldots,C_{t_i}^i$ for some $t_i\geq 0$. Assume that $\lambda^{\sf w}(C_1^i)\leq\ldots\leq\lambda^{\sf w}(C_{t_i}^i)$. For $j\in\{1,\ldots,t_i\}$, we chose the maximum $\lambda_{j}^i\in\Lambda$ such that $\lambda^{\sf w}(C_j^i)\geq\lambda_{j}^i$. Then we define $\Lambda^i=\langle\lambda_1^i,\ldots,\lambda_{t_i}^i\rangle$. It is straightforward to verify that (i)--(iii) are fulfilled for these values of $h_1,\ldots,h_k$, $p_1,\ldots,p_k$, $t_1,\ldots,t_k$, and $t_i$-tuples $\Lambda^i$ constructed for $i\in\{1,\ldots,k\}$. Also we have that   $\Lambda\preceq\sum_{i=1}^k((s_i-p_i)\langle \lambda\rangle+\Lambda^1+\ldots+\Lambda^p)$.

 Suppose now that there are  nonnegative integers $h_1,\ldots,h_k$, $p_1,\ldots,p_k$ and $t_1,\ldots,t_k$, and $t_i$-subtuples $\Lambda^i\subseteq\Lambda$ 
such that (i)--(iii) are fulfilled and  $\Lambda\preceq\sum_{i=1}^k((s_i-p_i)\langle \lambda\rangle+\Lambda^1+\ldots+\Lambda^p)$. Let $F_i\subseteq E(H_i)$ be a solution for $(H_i,w,(r_i-p_i)\langle i\rangle+\Lambda^i,h_i)$. Let $F=F_1\cup\ldots\cup F_k$. Because of (i), ${\sf w}(F)\leq k$, and by (ii) and (iii), $|\sum_{i=1}^k((s_i-p_i)\langle \lambda\rangle+\Lambda^1+\ldots+\Lambda^p)|=t$. Since $\Lambda\preceq\sum_{i=1}^k((s_i-p_i)\langle \lambda\rangle+\Lambda^1+\ldots+\Lambda^p)$, we conclude that $F$ is a solution for $(G,w,\Lambda,k)$, that is, this is a yes-instance of \probWCGC.
\end{subproof}

We use Claim~B to construct the brute force algorithm for  \probWCGC. We consider all $k$-tuples of nonnegative integers $\mathbf{h}=\langle h_1,\ldots,h_k\rangle$, $\mathbf{p}=\langle p_1,\ldots,p_k\rangle$ and $\mathbf{t}=\langle t_1,\ldots,t_k\rangle$ such that 
\begin{itemize}
\item[(i)] $h_1+\ldots+h_k\leq k$,
\item[(ii)] $p_i\leq s_i$, $p_i\leq t_i\leq 2h_i$ for $i\in\{1,\ldots,k\}$,
\item[(iii)]  $\sum_{i=1}^k(t_i-p_i)=t-s$.
\end{itemize}
Then we construct all possible $k$-tuples $\mathcal{L}=\langle \Lambda^1,\ldots,\Lambda^k\rangle$ such that $\Lambda^i$ is a $t_i$-subtuple of $\Lambda$ for $i\in\{1,\ldots,k\}$ and
$\Lambda\preceq\sum_{i=1}^k((s_i-p_i)\langle \lambda\rangle+\Lambda^1+\ldots+\Lambda^p)$.
For each choice of $(\mathbf{h},\mathbf{p},\mathbf{t},\mathcal{L})$, we verify whether $(H_i,w,(s_i-p_i)\langle i\rangle+\Lambda^i,h_i)$ is a yes-instance of \probWCGC for $i\in\{1,\ldots,k\}$ using Lemma~\ref{lem:uniform}. If it holds, we return {\sf YES} and stop. Otherwise, we discard the current choice of   $(\mathbf{h},\mathbf{p},\mathbf{t},\mathcal{L})$. 
If we do not return {\sf YES} for any choice of  $(\mathbf{h},\mathbf{p},\mathbf{t},\mathcal{L})$, we return {\sf NO} and stop.

Claim~B immediately implies correctness. To evaluate the running time of this part of the algorithm, observe that 
because of (i) and (ii), there are $2^{\Oh(k\log k)}$ choices of $\mathbf{h}$, $\mathbf{p}$ and $\mathbf{t}$. 
Since $\mathbf{var}(\lambda)\leq 3k$, there are $2^{\Oh(k\log k)}$ choices of each $\Lambda^i$ and $2^{\Oh(k^2\log k)}$ choices of each $\mathcal{L}$. 
Clearly,  the condition  $\Lambda\preceq\sum_{i=1}^k((s_i-p_i)\langle \lambda\rangle+\Lambda^1+\ldots+\Lambda^p)$ can be checked in polynomial time.
We verify whether 
 $(H_i,w,(s_i-p_i)\langle i\rangle+\Lambda^i,h_i)$ is a yes-instance of \probWCGC for $i\in\{1,\ldots,k\}$ in time $2^{2^{2^{2^{2^{2^{\Oh(k)}}}}}}\cdot n^{\Oh(1)}$ by Lemma~\ref{lem:uniform}.
Therefore, the total running time is $2^{2^{2^{2^{2^{2^{\Oh(k)}}}}}}\cdot n^{\Oh(1)}$.

Since preliminary steps of our algorithm including the application Reduction Rule~\ref{red:remove} are polynomial, the running time of the algorithm is $2^{2^{2^{2^{2^{2^{\Oh(k)}}}}}}\cdot n^{\Oh(1)}$.
\end{proof}

\section{Conclusion}\label{sec:concl}
We proved that \probCGClong is \classFPT when parameterized by $k$. We obtained this result by applying the recursive understanding technique~\cite{ChitnisCHPP16}. The drawback of this approach is that the dependence of the running time on the parameter is huge and it seems unlikely that using the same approach one could avoid towers of the exponents similar to the function in Theorem~\ref{thm:general}. In particular, we do not see how to avoid using mimicking networks (see~\cite{HagerupKNR98,KhanR14} for the definitions and lower and upper bounds for the size of such networks) whose sizes are double-exponential in the number of terminals.  Hence, our result is qualitative. The natural question would be to ask whether we can get a better running time using different techniques. This question is interesting even for some special cases of \probCGC when the connectivity constraints are bounded by a constant or are the same for all components. 
From the other side, it is natural to ask about lower bounds on the running time. 
For an \classFPT parameterized problem, it is natural to ask whether it admits a polynomial kernel. We observe that it is unlikely that \probWCGC has  a polynomial kernel even  if there are no weights and the maximum connectivity constraint is one,  because it was shown by Cygan et al. in~\cite{CyganKPPW14} that already \textsc{$t$-Cut} parameterized by the solution size $k$ has no polynomial kernel unless $\classNP\subseteq\classCoNP/{\rm poly}$. 
%We also proved in Theorem~\ref{thm:nokern} that it is unlikely that \probCGC parameterized by $k$ has a polynomial kernel but this result does not rule out the existence of a \emph{Turing kernel} (we refer to~\cite{CyganLPPS14} for the introduction) of polynomial size. Can it happen that a polynomial Turing kernel exists?  
Another direction of research is to consider vertex connectivities instead of edge connectivities.

%
%\bibliographystyle{siam}
%%\bibliographystyle{plain}
%\bibliography{complete}
 \newcommand{\bibremark}[1]{\marginpar{\tiny\bf#1}}
  \newcommand{\biburl}[1]{\url{#1}}

\end{document}